%% file: ms.tex
\documentclass[english]{article}
\pdfoutput=1
\usepackage[utf8]{inputenc} 
\usepackage[english]{babel} 
\usepackage[noend]{algpseudocode}
\usepackage[hyphens]{url}
\usepackage{hyperref}
\usepackage{cite}
\usepackage{graphicx} 
\usepackage{caption}
\usepackage{subcaption}
\usepackage{amsmath,amssymb}
\usepackage{float}
\usepackage{todonotes}
\usepackage[slant]{complexity}
\usepackage{thm-restate}
\usepackage{xspace}
\usepackage{microtype}
\usepackage{pdfpages}
\usepackage{color}
\usepackage{amsthm}
\usepackage{authblk}
\usepackage{fullpage}
\newtheorem{definition}{Definition}
\newtheorem{theorem}[definition]{Theorem}
\newtheorem{lemma}[definition]{Lemma}
\newtheorem{corollary}[definition]{Corollary}

\newcommand{\maxtap}{\textsc{MaxTAP}\xspace}
\newcommand{\tap}{\textsc{TAP}\xspace}
\newcommand{\remove}[1]{}
\date{}

\title{Tilt Assembly: Algorithms for Micro-Factories That Build Objects with Uniform External Forces}

\author[1]{Aaron T.~Becker\footnote{{Work from this author was partially supported by National Science Foundation IIS-1553063 and  IIS-1619278}}}
\author[2]{Sándor P.~Fekete}
\author[2]{Phillip Keldenich}
\author[2]{Dominik Krupke}
\author[2]{Christian Rieck}
\author[2]{Christian Scheffer}
\author[2]{Arne Schmidt}
\affil[1]{Department of Electrical and Computer Engineering, University of Houston, USA.
\tt atbecker@uh.edu}
\affil[2]{Department of Computer Science, TU Braunschweig, Germany.
	\tt$\{$s.fekete, p.keldenich, d.krupke, c.rieck, c.scheffer, arne.schmidt$\}$@tu-bs.de}

\newcommand{\revised}[1]{{\color{black} #1}}
\newcommand{\newlyrevised}[1]{{\color{black} #1}}
\newcommand{\revision}[1]{{\color{black}#1}}

\begin{document}
	\maketitle
	\begin{abstract}
		
We present algorithmic results for the parallel assembly of many micro-scale objects 
in two  and three dimensions from tiny particles, 
which has been proposed in the context of programmable matter and self-assembly	
for building high-yield micro-factories.
The underlying model has particles moving under the influence of uniform external forces
until they hit an obstacle; particles can bond when being forced together with
another appropriate particle.

Due to the physical and geometric constraints, not all shapes can be built in this manner; 
this gives rise to the {\sc Tilt Assembly Problem} (TAP) of deciding constructibility.
For simply-connected polyominoes $P$ in 2D consisting of $N$ unit-squares (``tiles''), we prove that TAP can	be
decided in $O(N\log N)$ time. 
For the optimization variant {\sc MaxTAP} (in which the
objective is to construct a subshape of maximum possible size), we show {\em polyAPX}-hardness: 
unless \P=\NP, {\sc MaxTAP} cannot be approximated within a factor of $\Omega(N^{\frac{1}{3}})$; for tree-shaped structures,
we give an $O(N^{\frac{1}{2}})$-approximation algorithm. 
For the efficiency of the
assembly process itself, we show that any constructible shape allows {\em
pipelined} assembly, which produces copies of $P$ in $O(1)$ amortized time,
i.e., $N$ copies of $P$ in $O(N)$ time steps. These considerations can be extended to three-dimensional objects: For the class of polycubes $P$ we prove that it is \NP-hard to decide whether it is possible to construct a path between two points
of $P$; it is also \NP-hard to decide constructibility of a polycube $P$. Moreover, it is 
{\em expAPX}-hard to maximize a path from a given start point.
	\end{abstract}

	\input{01-introduction}

	\input{02-preliminaries}
	\input{03-constructibility-simple}

	\input{4-short}
	\input{5-short}

	\input{06-future_work}

	\bibliography{bibliography}
\end{document}

%% file: 01-introduction.tex
\section{Introduction}

In recent years, progress on flexible construction at micro- and nano-scale
has given rise to a large set of challenges that deal with
algorithmic aspects of programmable matter. 
Examples of cutting-edge application areas with a strong
algorithmic flavor include self-assembling systems, in which chemical and
biological substances such as DNA are designed to form predetermined shapes or
carry out massively parallel computations; and swarm robotics, in which complex
tasks are achieved through the local interactions of robots with highly limited
individual capabilities, including micro- and nano-robots.

\remove{One particular difficulty when trying to assemble tiny particles to an overall
shape is to move the individual components to their appropriate locations,
where they can be attached, as individual navigation of tiny robotic devices suffers
from lack of energy and control. }
Moving individual particles to their appropriate attachment locations when assembling a shape is difficult because the small size of the particles limits the amount of onboard energy and computation.
One successful approach to dealing with this challenge is to use molecular diffusion in combination with cleverly designed
sets of possible connections: in {\em DNA tile self-assembly}, the particles are equipped
with sophisticated bonds that ensure that only a predesigned shape is produced
when mixing together a set of tiles\revision{, see~\cite{winfree1998algorithmic}}. 
The resulting study of algorithmic tile self-assembly
has given rise to an extremely powerful framework and produced a wide range of impressive results. 
However, the required properties of the building material (which must be 
specifically designed and finely tuned for each particular shape) in combination with the construction
process (which is left to chemical reactions, so it cannot be controlled or stopped until it
has run its course) make DNA self-assembly unsuitable for some applications.

An alternative method for controlling the eventual position of particles is to apply a 
uniform external force, causing all particles to move in a given direction until they hit an obstacle
or another blocked particle. As two of us (Becker and Fekete,~\revision{\cite{Becker3810a}}) have shown in the past, 
combining this approach with custom-made obstacles (instead of custom-made particles) 
allows complex rearrangements of particles, even in grid-like environments with axis-parallel motion. 
The appeal of this approach is that it shifts
the design complexity from the building material (the tiles) to the machinery (the environment).
As recent practical work by Manzoor et al.~\cite{manzoor2017parallel} shows, it is possible to apply this 
to simple ``sticky'' particles that can be forced to bond, see Fig.~\ref{fig:molecular}:
the overall assembly is achieved by adding particles one at a time, attaching them
to the existing sub-assembly. 

	\begin{figure}[h!]
		\centering
		\resizebox{0.2\textwidth}{!}{\includegraphics{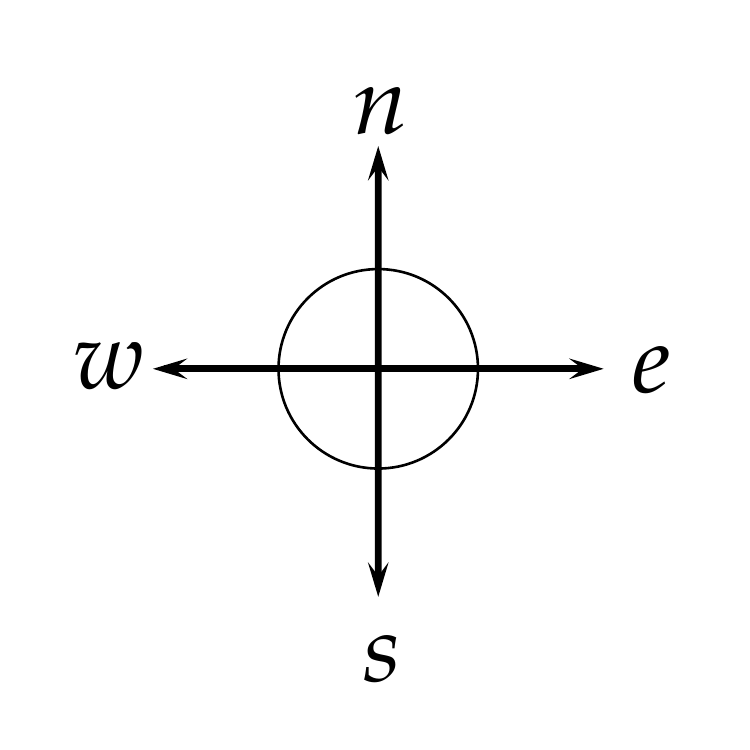}}
   \begin{minipage}[b]{0.75\textwidth}	
		\resizebox{\textwidth}{!}{\includegraphics{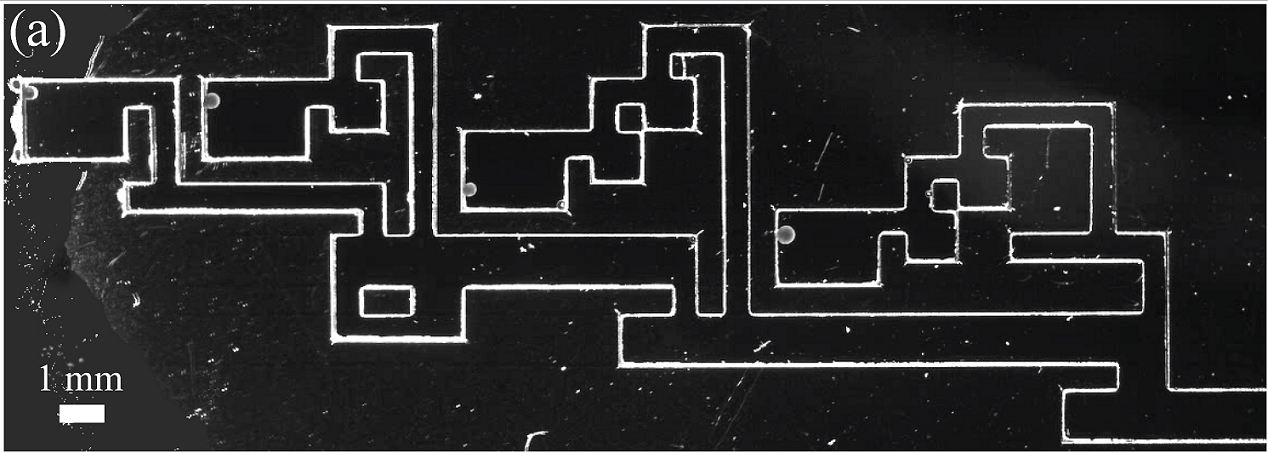}}
\\
		\resizebox{0.495\textwidth}{!}{\includegraphics{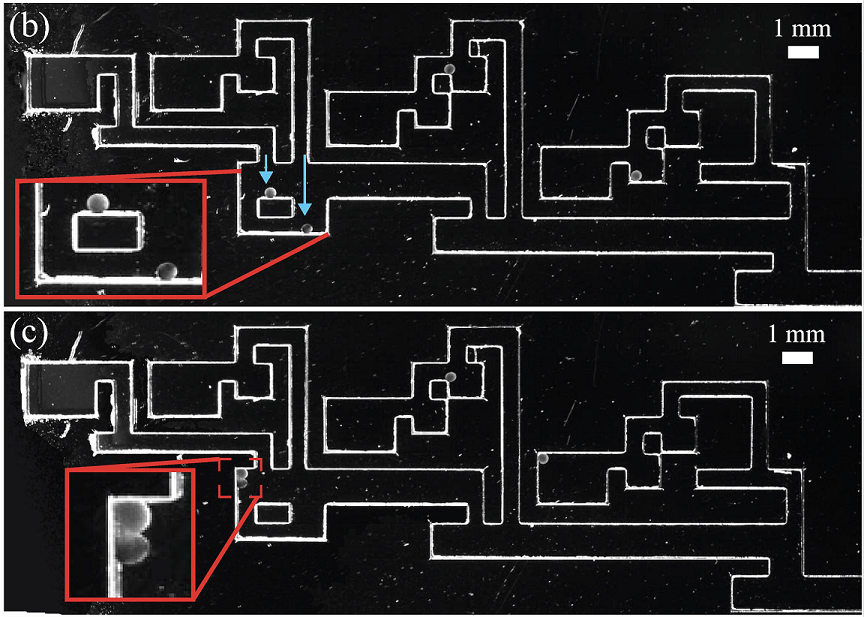}}
		\resizebox{0.495\textwidth}{!}{\includegraphics{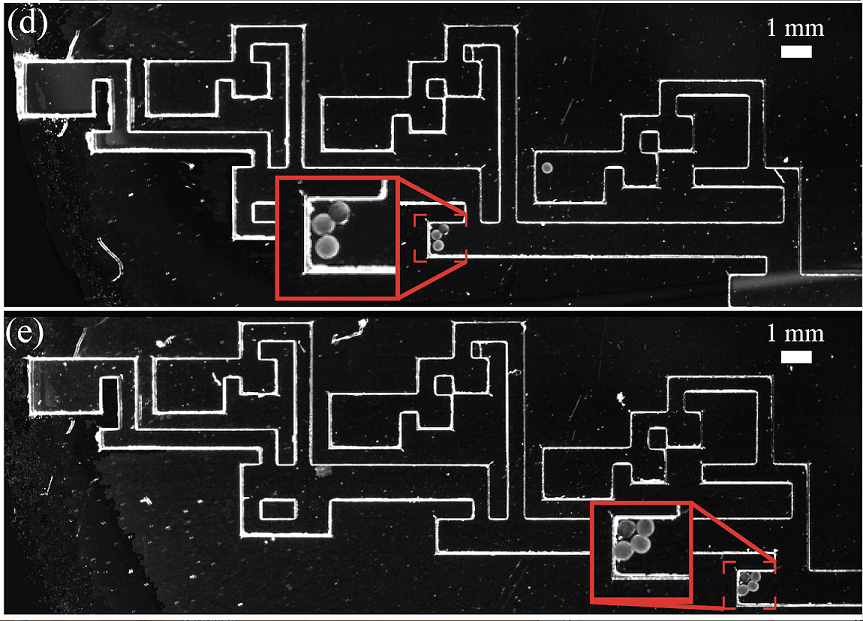}}
		  \end{minipage}
		\caption{A practical demonstration of Tilt Assembly based on alginate (i.e., a gel made by combining a powder derived from seaweed with water) particles~\cite{manzoor2017parallel}.
		(a)~Alginate particles in initial positions. (b)~After
control moves of $\langle e, s, w, n, e, s\rangle$ (for east, south, west, north), the alginate microrobots move to the
shown positions. (c)~After $\langle w, n\rangle$ inputs, the system produces the first
multi-microrobot polyomino. (d)~The next three microrobot polyominoes are produced after
applying multiple $\langle e,s,w,n\rangle$ cycles. (e)~After the alginate microrobots have
moved through the microfluidic factory layout, the final 4-particle polyomino
is generated.} 
\label{fig:molecular}
	\end{figure}

Moreover, pipelining \remove{the production}\revision{this} process may result in efficient rates \revision{of production},
see Fig.~\ref{fig:pipeline}~\cite{manzoor2017parallel}.

	\begin{figure}[h!]
		\centering
		\resizebox{0.1\textwidth}{!}{\includegraphics{./figures/CompassRose.pdf}}
\begin{minipage}[l]{0.45\textwidth}
		\resizebox{.99\textwidth}{!}{\includegraphics{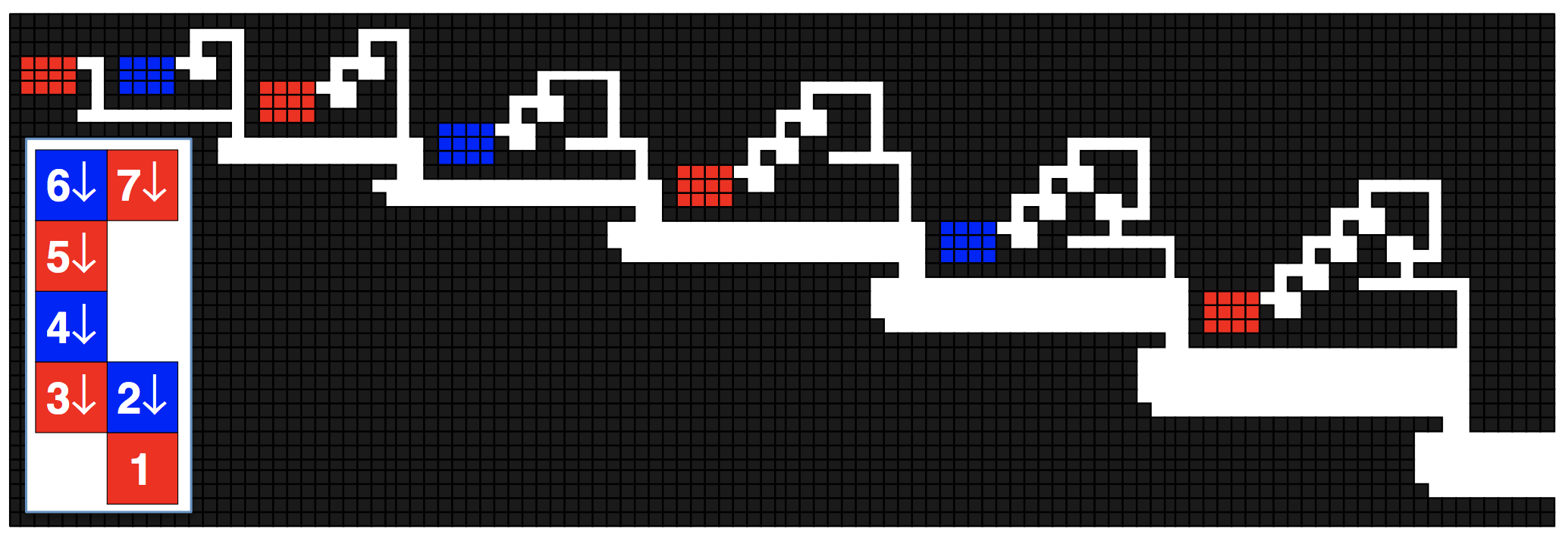}}
\\
		\resizebox{.99\textwidth}{!}{\includegraphics{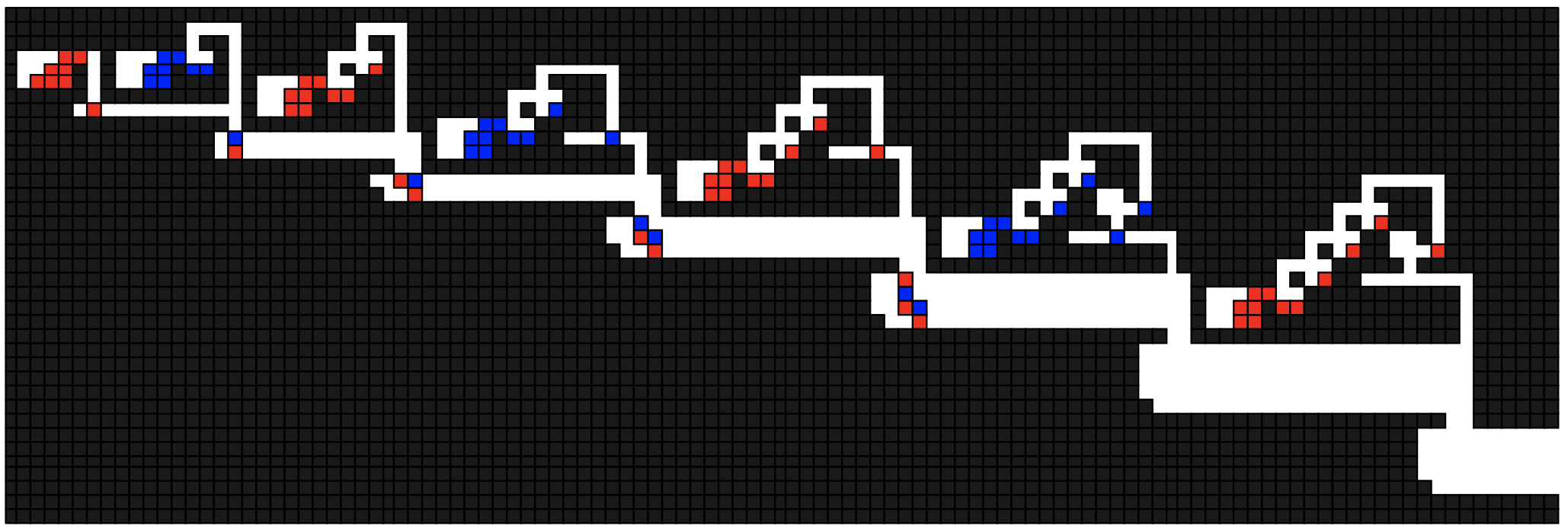}}
\end{minipage}
\begin{minipage}[l]{0.41\textwidth}
	\resizebox{0.99\textwidth}{!}{\includegraphics{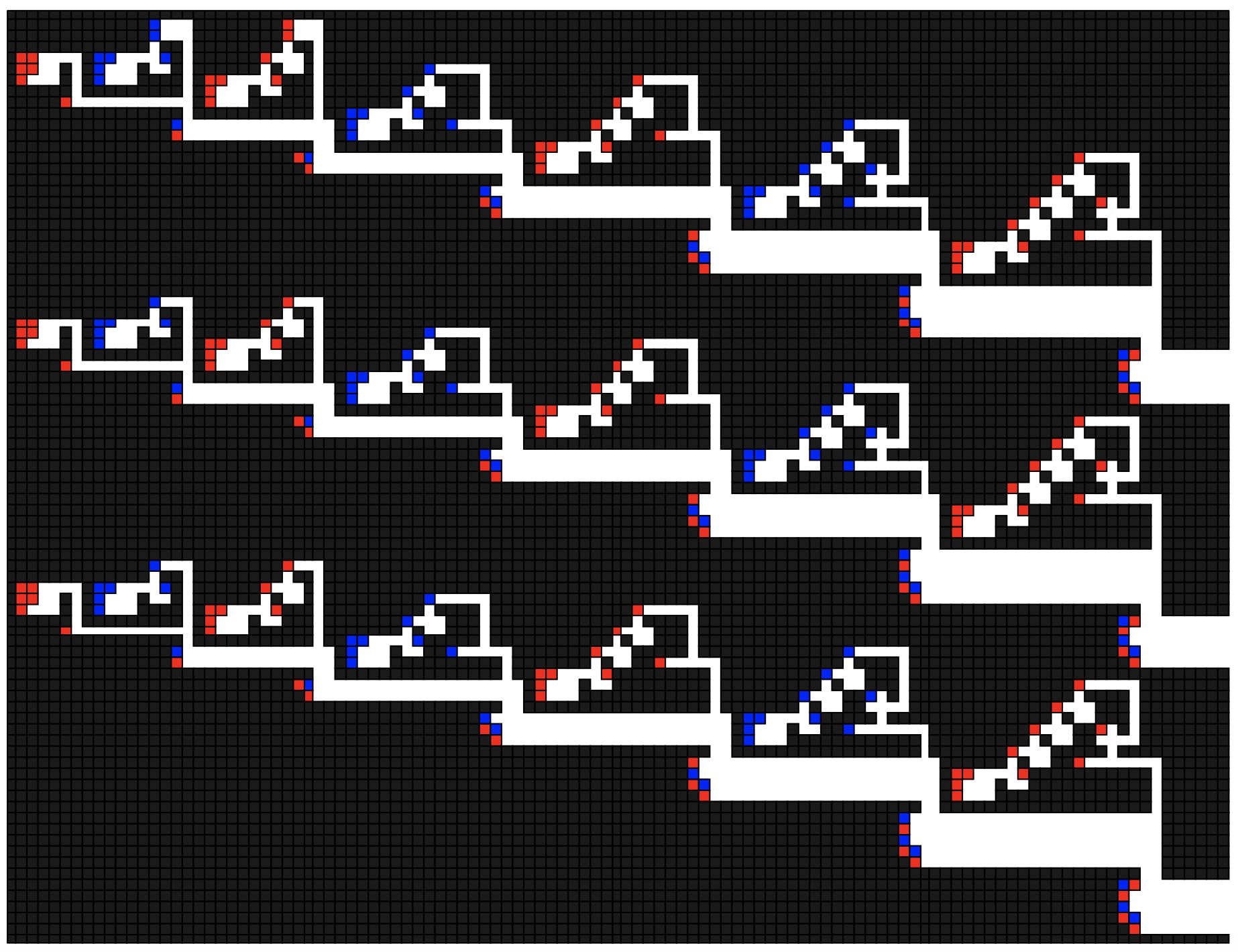}}
\end{minipage}
		\caption{(Top left) Initial setup of a seven-tile polyomino
assembly; the composed shape is shown enlarged on the lower left. The bipartite decomposition into blue and red particles is shown for greater clarity, but can also be used for better control of bonds. The sequence of control moves is $\langle e, s, w, n\rangle$, i.e., a clockwise order. (Bottom left) The situation after 18 control moves. (Right) The situation after 7 full cycles,
i.e., after 28 control moves; shown are three parallel ``factories''.}
\label{fig:pipeline}
	\end{figure}

One critical issue of this approach is the requirement of getting particles to their 
destination without being blocked by or bonding to other particles. As Fig.~\ref{fig:nope} shows,
this is not always possible, so there are some shapes that cannot be constructed
by Tilt Assembly.

	\begin{figure}
		\centering
		\resizebox{0.25\textwidth}{!}{\includegraphics{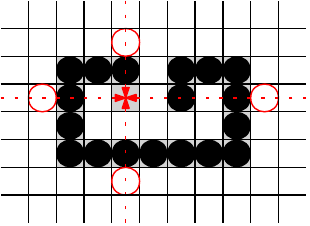}}
		\caption{A polyomino (black) that cannot be constructed by Tilt Assembly: the last tile cannot be attached, as it gets blocked by previously attached tiles.}
		\label{fig:nope}
	\end{figure}

This gives rise to a variety of algorithmic questions:
\revised{(1) Can we decide efficiently whether a given polyomino can be constructed by Tilt Assembly? 
(2) Can the resulting process be pipelined to yield low amortized building time? 
(3) Can we compute a maximum-size subpolyomino that can be constructed? 
(4)  What can be said about three-dimensional versions of the problem?}

	\subsection{Our Contribution}
	We present the results shown in Table~\ref{tab:results}.

\remove{	\begin{itemize}
		\item \tap is decidable in $O(N\log N)$ time for simple (i.e., hole-free) polyominoes \revised{(Section~\ref{sec:constructibility-simple})} \revision{but is \NP-hard for polycubes (Section~\ref{sec:3d-shapes})}.
		\item Any constructible polyomino can be built in a pipelined process, resulting in an amortized construction time of $O(1)$ \revised{(Section~\ref{sec:constructibility-simple})}.
		\item In 3D it is \NP-hard to decide if there is a constructible path between two given points of a polycube shape $P$ \revised{(Section~\ref{sec:3d-shapes})}.
	\end{itemize}
	
	\noindent \revision{Furthermore we present the following results for the optimization variant.}
	\begin{itemize}
		\item \remove{The optimization variant of finding a maximum cardinality constructible subpolyomino is \textit{polyAPX}-hard, i.e., }\maxtap cannot be approximated within a factor of $N^{\frac 1 3}$, unless \P=\NP\, \revised{(Section~\ref{sec:optimization-2d})}.
		\item There is an $N^{\frac 1 2}$-approximation for \maxtap if the polyomino $P$ is simple \revised{(Section~\ref{sec:optimization-2d})}.
		\item Maximizing a path from a given start point in a given polycube is \textit{expAPX}-hard \revised{(Section~\ref{sec:3d-shapes})}.
	\end{itemize}}
	{
	\renewcommand{\arraystretch}{1.5}	
	\setlength{\tabcolsep}{1.5mm}
	\begin{table}[h!]
	
	\begin{tabular}{|l|c|c|c|c|}
		\hline
		{\bf Dimension}&	\textbf{Decision}	&	\textbf{Maximization} & \textbf{Approximation} & \textbf{Constructible Path}\\
		\hline
		2D (simple)& $O(N\log N)$ (Sec.~\ref{sec:constructibility-simple})& \textit{polyAPX}-hard & $\Omega(N^{1/3})$, $O(\sqrt{N})$ (Sec.~\ref{sec:optimization-2d}) & $O(N\log N)$(Sec.~\ref{sec:optimization-2d})\\
			\hline
		3D (general) & \NP-hard \hfill(Sec.~\ref{sec:3d-shapes}) & \textit{polyAPX}-hard& $\Omega(N^{1/3})$, \hfill-\hfill (Sec.~\ref{sec:optimization-2d})& \NP-hard\hspace{0.3cm} (Sec.~\ref{sec:3d-shapes})\\\hline
	\end{tabular}
	\caption{Results for Tilt Assembly Problem (\tap) and its maximization variant (\maxtap)}
	\label{tab:results}
	\end{table}}


	\subsection{Related Work}
	Assembling polyominoes with tiles has been considered intensively in
the context of {\em tile self-assembly}. In 1998, Erik Winfree~\cite{winfree1998algorithmic}  
introduced the \textit{abstract tile self-assembly model} (aTAM), in which tiles have glue
types on each of the four sides and two tiles can stick together if their glue
type matches and the bonding strength is sufficient.
Starting with a \textit{seed tile}, tiles will continue to attach to the existing partial assembly until they form a desired polyomino; the process stops when no further attachments are possible.
	Apart from the aTAM, there are various other models like the
\textit{two-handed tile self-assembly model} (2HAM)~\cite{cannon2012two} and the
\textit{hierarchical tile self-assembly model}~\cite{chen2017parallelism}, in which 
we have no single seed but pairs of subassemblies that can attach to each other.
Furthermore, the \textit{staged self-assembly model}~\cite{demaine2008staged,demaine2017new} allows
greater efficiency by  assembling polyominoes in multiple bins which are gradually combined with the content of other bins. 
	
All this differs from the model	in Tilt Assembly, in which each tile has the same glue 
type on all four sides, and tiles are added to the assembly one at a time by attaching them
from the outside along a straight line. 
This approach of externally movable tiles has actually been considered in
practice at the microscale level using biological cells and an MRI, see
\cite{kim2015imparting}, \cite{kim2013swarm}, \cite{becker2014simultaneously}.
Becker et al.~\cite{becker2015toward} consider this for the assembly of a
magnetic {\em Gau{\ss} gun}, which can be used for applying strong local forces
by very weak triggers, allowing applications such as micro-surgery.
	
Using an external force for moving the robots becomes
inevitable at some scale because the energy capacity decreases faster than the
energy demand. A consequence is that all non-fixed robots/particles perform the same
movement, so all particles move in the same direction of the external force
until they hit an obstacle or another particle.  These obstacles allow shaping the
particle swarm.  Designing appropriate sets of obstacles and moves gives rise
to a range of algorithmic problems.  Deciding  whether a given initial
configuration of particles in a given environment can be transformed into a
desired target configuration is \NP-hard~\cite{Becker3810a}, even in a
grid-like setting, whereas finding an optimal control sequence is shown to be
\PSPACE-complete by Becker et al.~\cite{becker2014particle}. However, if 
it is allowed to {\em design} the obstacles in the first place, the problems
become much more tractable~\cite{Becker3810a}\revised{. Moreover, even} complex computations
\revised{become} possible: If we allow additional particles of double size (i.e., two adjacent fields), full
computational complexity is achieved, see Shad et
al.~\cite{shad2015particle}.
Further related work includes gathering a particle swarm at a single 
position~\cite{mahadev2016collecting} and using swarms of very simple robots
(such as Kilobots) for moving objects~\cite{becker2013massive}.
For the case in which human controllers have to move objects by such a swarm, Becker et
al.~\cite{becker2014crowdsourcing} study different control options.
The results are used by Shahrokhi and Becker~\cite{Shahrokhi2015}
to investigate an automatic controller.
	
Most recent and most closely related to our paper is the work by Manzoor et 
al.~\cite{manzoor2017parallel}, who use global control
to assembly polyominoes in a pipelined fashion: after constructing the
first polyomino, each cycle of a small control sequence produces another
polyomino. However, the algorithmic part is purely heuristic; providing
a thorough understanding of algorithms and complexity is the content of our paper.

%% file: 02-preliminaries.tex
\section{Preliminaries}
\begin{description}

\item[Polyomino:]
For a set $P\subset \mathbb{Z}^2$ of $N$ grid points in the plane, the graph
$G_P$ is the induced grid graph, in which two vertices $p_1, p_2\in P$ are
connected if they are at unit distance. Any set $P$ with connected grid graph
$G_P$ gives rise to a {\em polyomino} by replacing each point $p\in P$ by a unit 
square centered at $p$, which is called a {\em tile}; for simplicity, we also use $P$ to denote the
polyomino when the context is clear, and refer to $G_P$ as the dual graph of
the polyomino; $P$ is {\em tree-shaped}, if $G_P$ is a tree.
A polyomino is called {\em hole-free} or {\em simple} if and only if the grid graph
induced by $\mathbb{Z}^2\setminus P$ is connected. \\

\item [Blocking sets:] For each point $p\in \mathbb{Z}^2$ we define \textit{blocking sets} $N_p$, $S_p\subseteq P$ as the set of all points $q\in P$ that are above or below $p$ and $|p_x - q_x| \leq 1$. 
Analogously, we define the blocking sets $E_p$, $W_p\subseteq P$ as the set of all points $q\in P$ that are to the right or to the left of $p$ and $|p_y - q_y| \leq 1$.\\

\item[Construction step:]
A \textit{construction step} is defined by a direction (north, east, south,
west, abbreviated by $n,e,s,w$) from which a tile is added and a
latitude/longitude $l$ describing a column or row.  The tile
arrives from $(l, \infty)$ for north, $(\infty, l)$ for east, $(l,-\infty)$ for south, and $(-\infty, l)$ for west into the corresponding direction until it
reaches the first grid position that is adjacent to one occupied by an existing tile.
If there is no such tile, the polyomino does not change.  We note that a position $p$ can
be added to a polyomino $P$ if and only if there is a point $q\in P$ with
$||p-q||_1 = 1$ and one of the four blocking sets, $N_p$, $E_p$, $S_p$ or
$W_p$, is empty. \revision{Otherwise, if none of these sets are empty, this position is \emph{blocked.}}

\item[Constructibility:]
Beginning with a seed tile at some position $p$, 
a polyomino $P$ is {\em constructible} if and only if there is a sequence $\sigma =
\left((d_1,l_1),(d_2,l_2),\dots ,(d_{N-1},l_{N-1})\right)$, such that the
resulting polyomino $P'$, induced by successively adding tiles with $\sigma$,
is equal to $P$.
We allow the constructed polyomino $P'$ to be a translated copy of $P$.
\revision{Reversing $\sigma$ yields a \emph{decomposition sequence}, i.e., a sequence of tiles getting removed from $P$.}
\end{description}

%% file: 03-constructibility-simple.tex
\section{Constructibility of Simple Polyominoes}
\label{sec:constructibility-simple}
	In this section we focus on hole-free (i.e., simple) polyominoes. 
	We show that the problem of deciding whether a given polyomino can be constructed can be solved in polynomial time.
	This decision problem can be defined as follows.
	
	\begin{definition}[{\sc Tilt Assembly Problem}]
		Given a polyomino $P$, the {\sc Tilt Assembly Problem} (\tap) asks for a sequence 
of tiles constructing $P$, if $P$ is constructible.
	\end{definition}

\subsection{A Key Lemma}
	A simple observation is that construction and (restricted) decomposition are the same problem.
	This allows us to give a more intuitive argument, as it is easier to argue that we do not 
lose connectivity when removing tiles than it is to prove that we do not block future tiles.
	\begin{theorem}
		A polyomino $P$ can be constructed if and only if it can be decomposed \revision{using a sequence of} tile removal steps that \revised{preserve} connectivity.
		A construction sequence is a reversed decomposition sequence.
	\end{theorem}
	\begin{proof}
		\remove{We can limit on a single step: A tile removed in such a way can be added again, reversing the removal and vice versa.
		Thus, we can go back and forth in a construction/decomposition sequence.}
		To prove this theorem, it suffices to consider a single step.
		Let $P$ be a polyomino and $t$ be a tile that is removed from $P$ into some direction $l$, leaving a polyomino $P'$. 
		Conversely, adding $t$ to $P'$ from direction $l$ yields $P$, as there 
cannot be any tile that blocks $t$ from reaching the correct position, or 
		we would not be able to remove $t$ from $P$ in direction $l$.
	\end{proof}

	For hole-free polyominoes we can efficiently find a construction/decomposition sequence if one exists.
	The key insight is that one can greedily remove \emph{convex} tiles. \remove{tiles having as one of their corners a convex corner of the polyomino}\revision{A tile $t$ is said to be convex if and only if there is a $2\times 2$ square solely containing $t$}; see Fig.~\ref{fig:constr:removingwrongconvex}.
	If a convex tile is {\em not} a cut tile, i.e., it is a tile whose removal does {\em not} disconnect the polyomino, its removal does not interfere with the decomposability of the remaining polyomino.

	\begin{figure}
		\centering
		\begin{subfigure}[t]{0.42\textwidth}
			\centering
			\includegraphics{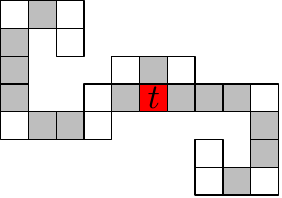}
			\caption{Removing $t$ destroys decomposability. The polyomino can be decomposed by starting with the three tiles above $t$.}
			\label{fig:constr:sop:nonconvex}
		\end{subfigure}
		\hfill
		\begin{subfigure}[t]{0.42\textwidth}
			\centering
			\includegraphics[scale=0.5]{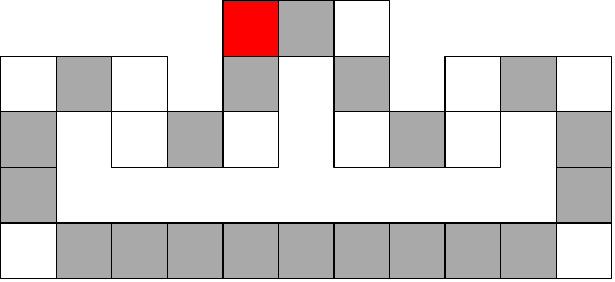}
			\caption{Removing the red convex tile leaves the polyomino non-decomposable; it can be decomposed by starting from the bottom or the sides.}
			\label{fig:constr:removingwrongconvex-nonsimple}
		\end{subfigure}
		\caption{\revised{Two polyominoes and their convex tiles (white).} (a) Removing non-convex tiles may destroy decomposability. (b) In case of non-simple polygons we may not be able to remove convex tiles.}
		\label{fig:constr:removingwrongconvex}
	\end{figure}

	This \revision{conclusion} is based on the observation that a minimal cut (i.e., a minimal set of vertices whose removal leaves a disconnected polyomino) of cardinality two in a hole-free polyomino \revision{always} consists of two (possibly diagonally) adjacent tiles.
	Furthermore, we can always find such a removable convex tile in any decomposable hole-free polyomino.
	This allows us to devise a simple greedy algorithm.

	We start by showing that if we find a non-blocked convex tile that is not a cut tile, we can simply remove it.
	It is important to focus on convex tiles, as the removal of non-convex tiles can harm the decomposability: see Fig.~\ref{fig:constr:sop:nonconvex} for an illustration.
	In non-simple polyominoes, the removal of convex tiles can destroy decomposability, as demonstrated in Fig.~\ref{fig:constr:removingwrongconvex-nonsimple}.

	\begin{lemma}
		\label{lemma:constr:sop:stillconstr}
		Consider a non-blocked non-cut convex tile $t$ in a hole-free polyomino $P$.
		The polyomino $P-t$ is decomposable if and only if $P$ is decomposable.
	\end{lemma}

	\begin{proof}
		The first direction is trivial: if $P-t$ is decomposable, \remove{then also }$P$ is decomposable \revised{as well}, because we can \remove{easily }remove \revised{the non-blocked tile} $t$ first \remove{(not blocked) }and afterwards use the \revised{existing} decomposition sequence for $P-t$.
		\remove{For the other direction we need some case distinctions.}\revised{The other direction requires some case distinctions.}
		\remove{We prove it by contradiction, assuming that $P$ is decomposable but $P-t$ is not, i.e., $t$ is an important tile for the later decomposition.}\revised{Suppose for contradiction that $P$ is decomposable but $P-t$ is not, i.e., $t$ is important for the later decomposition.}
	
	Consider a valid decomposition sequence for $P$ and the first tile $t'$ we cannot remove if we \revised{were to} remove $t$ in the beginning.
	W.l.o.g., let $t'$ be the first tile in this sequence (removing all previous tiles obviously does not destroy the decomposability).
	When we remove $t$ first, we are missing a tile, hence $t'$ cannot be blocked but has to be a cut tile in the remaining polyomino $P-t$.
	The presence of \remove{tile }$t$ preserves connectivity, i.e., $\{t,t'\}$ is a minimal cut on $P$.
	Because \remove{we do not allow any holes}\revised{$P$ has no holes}, then \remove{only possibility is that }$t$ and $t'$ \revised{must be} diagonal neighbors, sharing the neighbors $a$ and $b$.
	Furthermore, by definition neither of $t$ and $t'$ is blocked in some direction.
	We make a case distinction on the relation of these two directions.
	\begin{description}
		\item[The directions are orthogonal (Fig.~\ref{fig:constr:sop:orthogonal}).]
		Either $a$ or $b$ is a non-blocked convex tile, because $t$ and $t'$ are both non-blocked\revised{; w.l.o.g., let this be $a$}.
		\remove{Let us assume it is $a$.}%
		It is easy to see that independent of removing $t$ or $t'$ first, after removing $a$ we can also remove the other one.
		\item[The directions are parallel (Fig.~\ref{fig:constr:sop:parallel}).]
		This case is slightly more involved.
		By assumption, we have a decomposition sequence beginning with $t'$.
		We show that swapping $t'$ with our convex tile $t$ in this sequence preserves feasibility.
		
		The original sequence has to remove either $a$ or $b$ before it removes $t$, as otherwise the connection between the two is lost when $t'$ is removed first.
		After either $a$ or $b$ is removed, $t$ becomes a leaf and can no longer be important for connectivity.
		Thus, we only need to consider the sequence until either $a$ or $b$ is removed.
		The main observation is that $a$ and $b$ block the same tiles as $t$ or $t'$, except for tile $c$ as in Fig.~\ref{fig:constr:sop:parallel}.
		However, when $c$ is removed, it has to be a leaf, because $a$ is still not removed and in the original decomposition sequence, $t'$ has already been removed.
		Therefore, a tile $d\neq t'$ would have to be removed before $c$.
		Hence, the decomposition sequence remains feasible, concluding the proof.\qedhere
	\end{description}
\end{proof}
	\begin{figure}
	\centering
	\begin{subfigure}[t]{0.45\textwidth}
		\centering
		\resizebox{0.7\columnwidth}{!}{\includegraphics{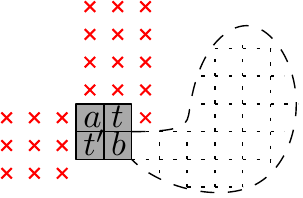}}
		\caption{If the unblocked directions of $t$ and $t'$ are orthogonal, one of the two adjacent tiles (w.l.o.g. $a$) \revised{cannot have any} further neighbors. There can also be no tiles in the upper left corner, because the polyomino cannot cross the two free directions of $t$ and $t'$ (red marks).}
		\label{fig:constr:sop:orthogonal}
	\end{subfigure}
	\hfil
	\begin{subfigure}[t]{0.45\textwidth}
		\centering
		\resizebox{0.7\columnwidth}{!}{\includegraphics{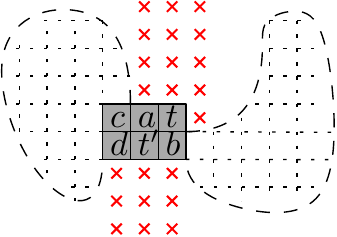}}
		\caption{If the unblocked directions of $t$ and $t'$ are parallel, there is only the tile $c$ for which something can change if we remove $t$ before $t'$.}
		\label{fig:constr:sop:parallel}
	\end{subfigure}
	\caption{The red marks indicate that no tile is at this position; the dashed outline represents the rest of the polyomino.}
\end{figure}

Next we show that such a convex tile always exists if the polyomino is decomposable.
\begin{figure}[h!]
	\centering
	\begin{subfigure}[t]{0.45\textwidth}
		\centering
		\resizebox{0.55\columnwidth}{!}{\includegraphics{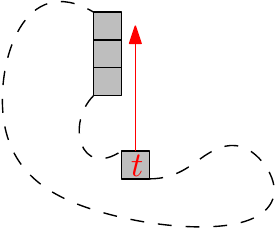}}
		\caption{If the removal direction of $t$ is not crossed, the last blocking tile has to be convex (and has to be removed before).}
		\label{fig:constr:sop:convex:a}
	\end{subfigure}
	\hfil
	\begin{subfigure}[t]{0.45\textwidth}
		\centering		
		\resizebox{0.65\columnwidth}{!}{\includegraphics{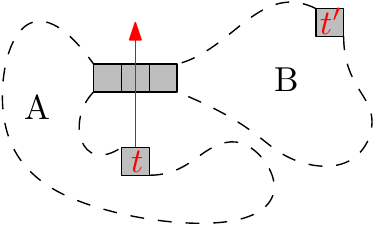}}
		\caption{\remove{If some tile is crossed by the removal direction of $t$, component $A$ is disconnected from component $B$, because the polygon is hole-free. But component $B$ also needs to have a convex tile $t'$ would have to be removed before.}
			If the removal direction of $t$ crosses $P$, then $P$ gets split into components $A$ and $B$. Component $B$ has a convex tile $t'$ that needs to be removed before $t$.}
		\label{fig:constr:sop:convex:b}
	\end{subfigure}
	\caption{\revised{Polyominoes} for which no convex tile should be removable, showing the contradiction to $t$ being the first blocked convex tile in $P$ getting removed.}
	\label{fig:constr:sop:convex}
\end{figure}
\begin{lemma}
	\label{lemma:constr:sop:removeconex}
	Let $P$ be a decomposable polyomino.
	Then there exists a convex tile that is removable without destroying connectivity.
\end{lemma}

\begin{proof}
	We prove this by contradiction based on two possible cases.	
	
	Assume $P$ to be a decomposable polyomino in which no convex tile is removable.
	Because $P$ is decomposable, there exists some feasible decomposition sequence $S$.
	Let $P_\text{convex}$ denote the set of convex tiles of $P$ and let $t\in P_\text{convex}$ be the first removed convex tile in the decomposition sequence $S$.
	By assumption, $t$ cannot be removed yet, so it is either blocked or a cut tile. 
	\begin{description}
		\item[$t$ is blocked.] Consider the direction in which we would remove $t$.
		If it does not cut the polyomino, the last blocking tile has to be convex (and would have to be removed before $t$), see Fig.~\ref{fig:constr:sop:convex:a}.
		If it cuts the polyomino, the component cut off also \revised{must} have a convex tile and the full component has to be removed before $t$, see Fig.~\ref{fig:constr:sop:convex:b}.
		This is again a contradiction to $t$ being the first convex tile to be removed in $S$.
		\item[$t$ is a cut tile.] $P-t$ consists of exactly two connected polyominoes,  $P_1$ and $P_2$.
		It is easy to see that $P_1\cap P_\text{convex}\not = \emptyset$ and $P_2\cap P_\text{convex}\not = \emptyset$, because every polyomino of size $n\geq 2$ has at least two convex tiles of which at most one becomes non-convex by adding $t$. (A polyomino of size $1$ is trivial.)	
		Before being able to remove $t$, either $P_1$ or $P_2$ has to be completely removed, including their convex tiles.
		This is a contradiction to $t$ being the first convex tile in $S$ to be removed.\qedhere
	\end{description}
\end{proof}

\subsection{An Efficient Algorithm}
An iterative combination of these two lemmas proves the correctness of greedily removing convex tiles.
As we show in the next theorem, using a search tree technique allows an efficient implementation of this greedy algorithm.
\begin{theorem}
	A hole-free polyomino can be checked for decomposability/constructibility in \revised{time} $O(N \log N)$. 
\end{theorem}
\begin{proof}
	Lemma~\ref{lemma:constr:sop:stillconstr} allows us to remove any convex tile, as long as it is not blocked and does not destroy connectivity.
	Applying the same lemma on the remaining polyomino iteratively creates a feasible decomposition sequence.
	Lemma~\ref{lemma:constr:sop:removeconex} proves that this is always sufficient.
	If and only if we can at some point no longer find a matching convex tile (to which we refer 
as {\em candidates}), the polyomino cannot be decomposable.
	
	Let $B$ be the time needed to check whether a tile $t$ is blocked.
	A na\"ive way of doing this is to try out all tiles and check if $t$ gets blocked,
	requiring time $O(N)$.
	With a preprocessing step, we can decrease $B$ to $O(\log N)$ by using $O(N)$ binary search trees for searching for blocking tiles and utilizing that removing a tile can change the state of at most $O(1)$ tiles.
	For every vertical line $x$ and horizontal line $y$ going through $P$, we create a balanced search tree, i.e., for a total of $O(N)$ search trees.
	An $x$-search tree for a vertical line $x$ contains tiles lying on \revised{$x$}, sorted by their $y$-coordinate.
	Analogously define a $y$-search tree for a horizontal line $y$ containing tiles lying on $y$ sorted by their $x$-coordinate.
	We iterate over all tiles $t=(x,y)$ and insert the tile in the corresponding $x$- and $y$-search tree with a total complexity of $O(N\log N)$.
	Note that the memory complexity remains linear, because every tile is in exactly two search trees.
	To check if a tile at position $(x', y')$ is blocked from above, we can simply search in the $(x'-1)$-, $x'$- and $(x'+1)$-search tree for a tile with $y>y'$.
	We analogously perform search queries for the other three directions, and thus have $12$ queries of total cost $O(\log N)$.
	
	We now iterate on all tiles and add all convex tiles that are not blocked and are not a cut tile to the set $F$ (cost $O(N \log N)$).
	Note that checking whether a tile is a cut tile can be done in constant time, because it suffices to look into the local neighborhood.
	While $F$ is not empty, we remove a tile from $F$, from the polyomino, and from its two search trees in time $O(\log N)$.
	Next, we check the up to $12$ tiles that are blocked first from the removed tile for all four orientations, see Fig.~\ref{fig:constr:tree}.
	\begin{figure}
		\centering
		\resizebox{0.25\textwidth}{!}{\includegraphics{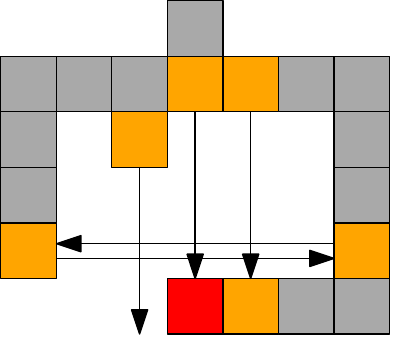}}
		\caption{When removing the red tile, only the orange tiles can become unblocked or convex.}
		\label{fig:constr:tree}
	\end{figure}
	Only these tiles can become unblocked or a convex tile.
	Those that are convex tiles, not blocked and no cut tile are added to $F$.
	All tiles behind those cannot become unblocked as the first tiles would still be blocking them.
	The cost for this is again in $O(\log N)$.
	This is continued until $F$ is empty, which takes at most $O(N)$ loops each of cost $O(\log N)$.
	If the polyomino has been decomposed, the polyomino is decomposable/constructible by the corresponding tile sequence.
	Otherwise, there cannot exist such a sequence.
	By prohibiting to remove a specific tile, one can force a specific start tile.	
\end{proof}

\begin{figure}
	\centering\hfill
	\resizebox{0.15\textwidth}{!}{\includegraphics{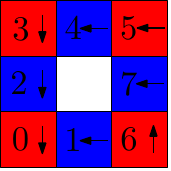}}\hfill
	\resizebox{0.5\textwidth}{!}{\includegraphics{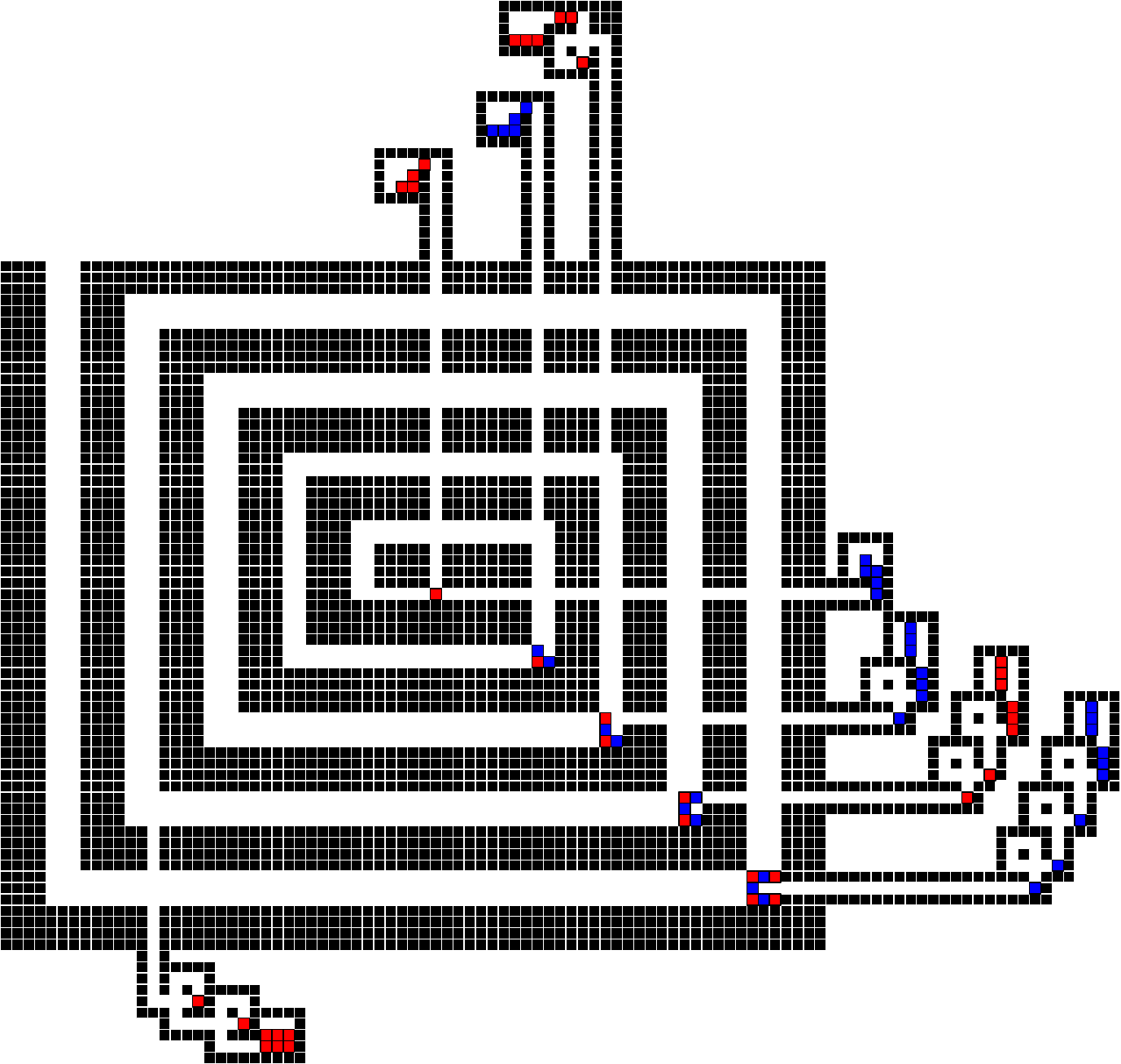}}
	\caption{(Left)	A polyomino $P$. Shown is the assembly order and the
direction of attachment to the seed (tile 0). (Right) A maze environment for pipelined construction
of the desired polyomino $P$. After the fourth cycle, each further
cycle produces a new copy of $P$.} \label{fig:parallel}
\end{figure}

\subsection{Pipelined Assembly}
Given that a construction is always possible based on adding convex corners to a partial construction,
we can argue that the idea of Manzoor et al.~\cite{manzoor2017parallel} for pipelined assembly can
be realized for {\em every} constructible polyomino:
We can transform the construction sequence into a spiral-shaped maze environment, as illustrated
in Fig.~\ref{fig:parallel}.
This allows it to produce $D$ copies of $P$ in $N+D$ cycles, implying that we only need $2N$ cycles 
for $N$ copies. 
It suffices to use a clockwise order of four unit steps (west, north, east, south)
in each cycle.

The main idea is to create a spiral in which the assemblies move from the inside to the outside.
The first tile is provided by an initial south movement.
After each cycle, ending with a south movement, the next seed tile of the next copy of $P$ is added.
For every direction corresponding to the direction of the next tile added by the sequence, 
we place a tile depot on the outside of the spiral, with a straight-line path to the location
of the corresponding attachment. 

\begin{theorem}\label{th:parallel}
	Given a construction sequence $\sigma :=\left((d_1,l_1),\dots
,(d_{N-1},l_{N-1})\right)$ that constructs a polyomino $P$, we can construct a
maze environment for pipelined tilt assembly, such that constructing $D$
copies of $P$ needs $O(N+D)$ unit steps. In particular, constructing one copy
of $P$ can be done in amortized time $O(1)$.  
\end{theorem}

\begin{proof}
	Consider the construction sequence $\sigma$, 
	the movement sequence $\zeta$ consisting of $N$ repetitions of the cycle ($w$, $n$, $e$, $s$), 
	and an injective function $m: \sigma \rightarrow \zeta$, with $m((w,\cdot)) = e$, $m((n,\cdot)) = s$, $m(e,\cdot)) = w$ and $m((s,\cdot)) = n$.
	We also require that $m((d_i,l_i)) = \zeta_j$ if for all $i' < i$ there is a $j' < j$ with $m((d_{i'},l_{i'})) = \zeta_{j'}$ and $j$ is smallest possible.
	This implies that in each cycle there is at least one tile in $\sigma$ mapped to one direction in this cycle.
	\begin{description}
		\item[Labyrinth construction:] The main part of the labyrinth is a spiral as can be seen in Fig.~\ref{fig:parallel}.
		Consider a spiral that is making $|\zeta|$ many turns, and the innermost point $q$ of this spiral.
		From $q$ upwards, we make a lane through the spiral until we are outside the spiral.
		At this point we add a depot of tiles, such that after each south movement a new tile comes out of the depot (this can easily be done with bottleneck constructions as seen in Fig.~\ref{fig:parallel} or in~\cite{manzoor2017parallel}).
		Then, we proceed for each turn in the spiral as follows:
		For the $j$-th turn, if $m^{-1}(\zeta_j)$ is empty we do nothing.
		Else if $ m^{-1}(\zeta_j)$ is not empty we want to add the next tile.
		Let $t_i$ be this particular tile.
		Then, we construct a lane in direction $-\zeta_j$, i.e., the direction from where the tile will come from, until we are outside the spiral.
		By shifting this line in an orthogonal direction we can enforce the tile to fly in at the correct position relating to $l_i$. 
		There, we add a depot with tiles, such that the first tile comes out after $j-1$ steps and with each further cycle a new tile comes out (this can be done by using loops in the depot, see Fig.~\ref{fig:parallel} or~\cite{manzoor2017parallel}).
		Depots, which lie on the same side of the spiral, can be shifted arbitrarily, so they do not collide.
		These depots can be made arbitrarily big, and thus, we can make as many copies of $P$ as we wish.
		Note that we can make the paths in the spiral big enough, such that after every turn the bounding box of the current polyomino fits through the spiral.
		
		\item[Correctness:] We will now show that we will obtain copies of $P$.
		Consider any $j$-th turn in the spiral, where the $i$-th tile $t_i$ is going to be added to the current polyomino.
		With the next step, $t_i$ and the polyomino move in direction $\zeta_j$.
		While the polyomino does not touch the next wall in the spiral, the distance between $t_i$ and the polyomino will not decrease.
		\remove{Thus, consider the situation, when the polyomino hits the wall:
		the polyomino stops moving and $t_i$ continues moving towards the polyomino.
		This is the same situation as in our non-parallel model:}However when the polyomino hits the wall the polyomino stops moving and ti  continues moving towards the polyomino.  Wall-hitting is the same situation as in our non-parallel model:
		To a fixed polyomino we can add tiles from $n$, $e$, $s$ or $w$.
		Therefore, the tile connects to the correct place.
		Since this is true for any tile and any copy, we conclude that every polyomino we build is indeed a copy of $P$.
		
		\item[Time:] Since the spiral has at most $4N$ unit steps (or $N$ cycles), the first polyomino will be constructed after $4N$ unit steps.
		By construction, we began the second copy one cycle after beginning the first copy, the third copy one cycle after the second, and so on.
		This means, after each cycle, when the first polyomino is constructed, we obtain another copy of $P$.
		Therefore, for $D$ copies we need $N+D$ cycles (or $O(N+D)$ unit steps). For $D\in \Omega(N)$ this results in an amortized constant time construction for $P$.
	\end{description}
	
	Note that this proof only considers construction sequences in the following form:
	If a tile $t_i$ increases the side length of the bounding box of the current polyomino, then the tile is added from a direction with a longitude/latitude, such that the longitude/latitude intersects the bounding box (see Fig.~\ref{fig:parallel:bad_sequence}).
	In the case there is a tile, such that the longitude/latitude does not intersect the bounding box, then we can rotate the direction by $\frac \pi 2$ towards the polyomino and we will have a desired construction sequence.
\end{proof}

\begin{figure}
	\centering
	\resizebox{0.4\textwidth}{!}{\includegraphics{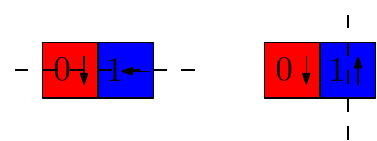}}
	\caption{Two different sequences. The red tile represents the bounding box of the current polyomino. (Left) A desired sequence. The latitude intersects the bounding box. (Right) A sequence where the latitude does not intersect the bounding box.}
	\label{fig:parallel:bad_sequence}
\end{figure}

%% file: 4-short.tex
\section{Optimization Variants in 2D}
	\label{sec:optimization-2d}
	For polyominoes that cannot be assembled, it is natural to look for a maximum-size subpolyomino that is constructible.
	This optimization variant is polyAPX-hard, i.e., we cannot hope for an approximation algorithm with an approximation factor within $\Omega(N^{\frac 1 3})$, unless $\P =\NP$.
	
	\begin{definition}[Maximum Tilt Assembly Problem]
		Given a polyomino $P$, the Maximum Tilt Assembly Problem (\maxtap) asks for a sequence of tiles building a cardinality-maximal connected subpolyomino $P'\subseteq P$.
	\end{definition}

	\begin{theorem}
	\label{th:polyAPX}
		\maxtap is polyAPX-hard, even for tree-shaped polyominoes.
	\end{theorem}

		\begin{proof}
			We reduce \textsc{Maximum Independent Set} (MIS) to \maxtap; see Fig.~\ref{fig:MIS} for an illustration.
			Consider an instance $G=(V,E)$ of MIS, which we transform into a polyomino $P_G$.
			We construct $P_G$ as follows.
			Firstly, construct a horizontal line from which we go down to select which vertex in $G$ will be chosen. 
			The line must have length $10n-9$, where $n = |V|$.
			Every 10th tile will represent a vertex, starting with the first tile on the line.
			Let $t_i$ be such a tile representing vertex $v_i$.
			For every $v_i$ we add a selector gadget below $t_i$ and for every $\{v_i,v_j\}\in\delta(v_i)$ we add a reflected selector gadget below $t_j$, as shown in Fig.~\ref{fig:MIS}, \revision{each consisting of 19 tiles}.
			Note that all gadgets for selecting vertex $v_i$ are above the gadgets of $v_j$ if $i < j$ \revision{and that there are at most $n^2$ such gadgets}.
			After all gadgets have been constructed, \revision{we have already placed at most $19n^2 + 10n-9\leq 29n^2$ tiles}. We continue with a vertical line with a length of $30n^2$ tiles. 
			\begin{figure}[h!]
				\centering
				\resizebox{0.4\columnwidth}{!}{\includegraphics{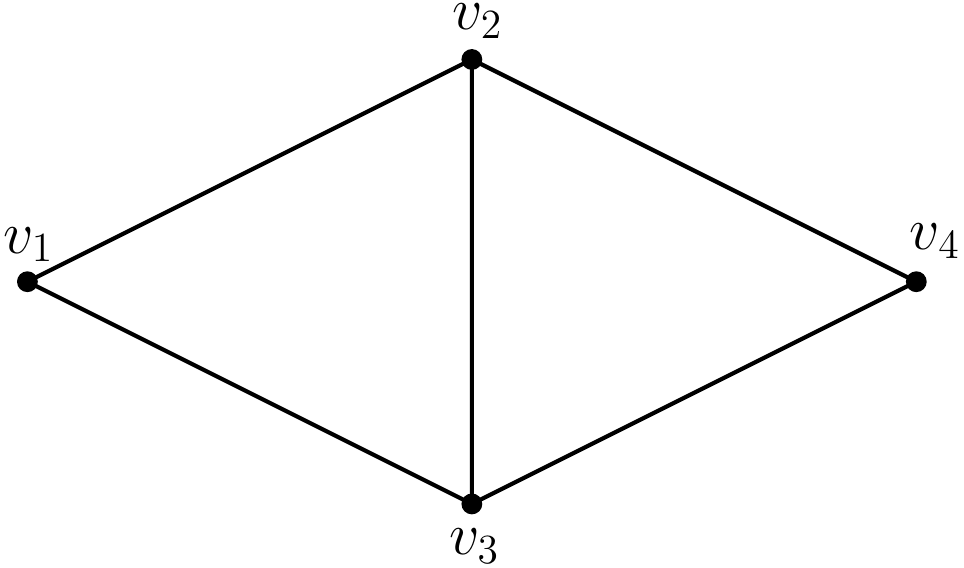}}
				\hfill
				\resizebox{0.49\columnwidth}{!}{\includegraphics{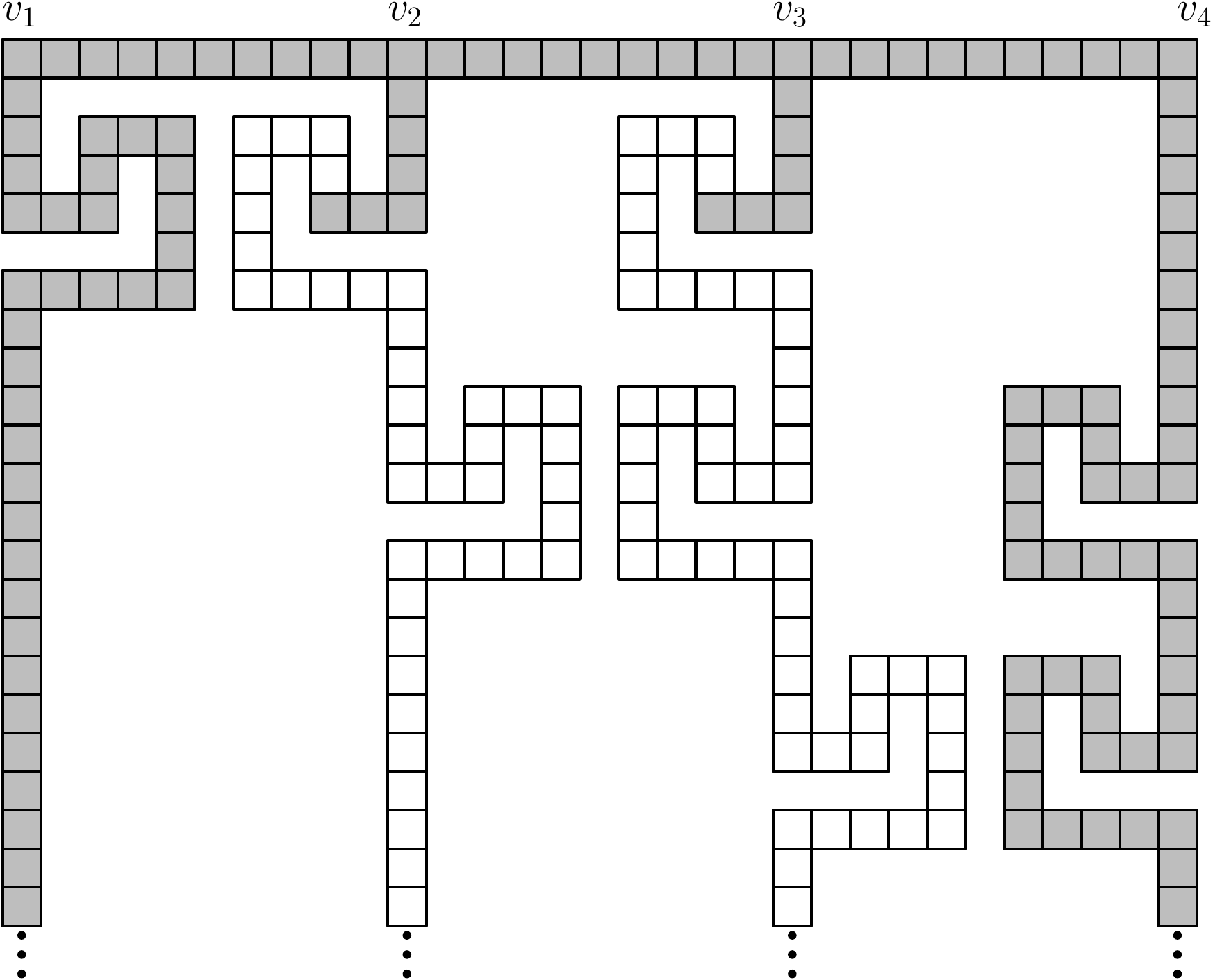}}
				\caption{Reduction from MIS to \maxtap. (Left) A graph $G$ with four vertices. (Right) A polyomino constructed for the reduction with a feasible, maximum solution marked in grey.}
				\label{fig:MIS}
			\end{figure}
			
			Now let $\alpha^*$ be an optimal solution to MIS.
			Then \maxtap has a maximum polyomino of size at least $30n^2\alpha^*$ and at most $30n^2\alpha^* + 29n^2$:
			We take the complete vertical part of $t_i$ for every $v_i$ in the optimal solution of MIS.
			Choosing other lines block the assembly of further lines and thus, yields a smaller solution.
			
			Now suppose we had an $N^{1-\varepsilon}$-approximation for \maxtap.
			Then we would have a solution of at least $\frac 1 {N^{1-\varepsilon}} T^*$, where $T^*$ is the optimal solution.
			We know that \revision{an optimal solution} has $T^*\geq 30n^2\alpha^*$ tiles and \revision{the polyomino has at most} $N \leq 30n^3 + 29n^2\leq 59n^3$ tiles.
			Therefore, we have at least  $\frac {30n^2\alpha^*}{59^{1-\varepsilon}n^{3-3\varepsilon}}$ tiles
			and thus at least
			$\frac {1}{59^{1-\varepsilon}n^{3-3\varepsilon}}\alpha ^*$ strips,
			because each strips is $30n^2$ tiles long.
			Consider some $\varepsilon \geq \frac 2 3 + \eta$ for any $\eta > 0$, then the number of strips is
			$\frac {1}{59^{1/3}n^{1-3\eta}}\alpha ^*$ which results in an $n^{1-\delta}$-approximation for MIS, contradicting the inapproximability of MIS (unless \P=\NP) \revision{shown by  Berman and Schnitger~\cite{berman1992complexity}}.
		\end{proof}

		As a consequence of the construction, we get Corollary~\ref{cor:13}.

	\begin{corollary}
	\label{cor:13}
		Unless $P=NP$, \maxtap cannot be approximated within a factor of $\Omega(N^{\frac 1 3})$.
	\end{corollary}

	On the positive side, we can give an $O(\sqrt{N})$-approximation algorithm.

	\begin{theorem}
		The longest constructible path in a tree-shaped polyomino $P$ is a $\sqrt{N}$-approximation for \maxtap, and we can find such a path in polynomial time.
	\end{theorem}

	\begin{proof}
		Consider an optimal solution $P^*$ and a smallest enclosing box $B$ containing $P^*$.
		Then there must be two opposite sides of $B$ having at least one tile of $P^*$.
		Consider the path $S$ between both tiles.
		Because the area $A_B$ of $B$ is at least the number of tiles in $P^*$, $|S| \geq \sqrt{A_B}$ and a longest, constructible path in $P$ has length at least $|S|$, we conclude that the longest constructible path is a $\sqrt N$-approximation.

		To find such a path, we can search for every path between two tiles, check whether we can build this path, and take the longest, constructible path.
	\end{proof}

	Checking constructibility for $O(N^2)$ possible paths is rather expensive.
	However, we can efficiently approximate the longest constructible path in a tree-shaped polyomino \revision{with the help of \emph{sequentially} constructible paths, i.e., the initial tile is a leaf in the final path}.
	\begin{theorem}\label{th:sequential:thin}
		We can find a constructible path in a tree-shaped polyomino in $O(N^2 \log N)$ time that has a length of at least half the length of the longest constructible path.
	\end{theorem}
	\begin{proof}
		We only search for paths that can be built sequentially.
		Clearly, the longest such path is at least half as long as the longest path that can have its initial tile anywhere.
		We use the same search tree technique as before to look for blocking tiles.
		Select a tile of the polyomino as the initial tile.
		Do a depth-first search and for every tile in this search, check if it can be added to the path.
		If it cannot be added, skip all deeper tiles, as they also cannot be added.
		During every step in the depth-first search, we only need to change a single tile in the search trees, doing $O(1)$ updates with $O(\log N)$ cost.
		As we only consider $O(N)$ vertices in the depth-first search, this results in a cost of $O(N \log N)$ for a fixed start tile.
		It is trivial to keep track of the longest such constructible path.
		Repeating this for every tile results in a running time of $O(N^2 \log N)$.
	\end{proof}
	
	In tree-shaped polyominoes, finding a constructible path is easy. 
	For simple polyominoes, additional arguments and data structures lead to a similar result.
	
	\begin{theorem}
	\label{th:N2}
	\revision{In simple polyominoes, finding the longest of all shortest paths that are sequentially constructible takes $O(N^2\log N)$ time.}
	
	\remove{The longest of all shortest paths that are sequentially constructible can be found in $O(N^2\log N)$ time.}
	\end{theorem}

\newlyrevised{Before we start with the proof of Theorem~\ref{th:N2}, we show in the next two lemmas that it is sufficient to consider shortest paths only, and that we can restrict ourselves to one specific shortest path between two tiles. Hence, we just need to test a maximum of $O(n^2)$ different paths.}

	\begin{lemma}
		\label{lem:sequential}
		In a sequentially constructible path, if there is a direct straight connection for a subpath, the subpath can be replaced by the straight connection.
	\end{lemma}
	
	\begin{figure}[h!]
		\centering
		\resizebox{0.3\columnwidth}{!}{\includegraphics{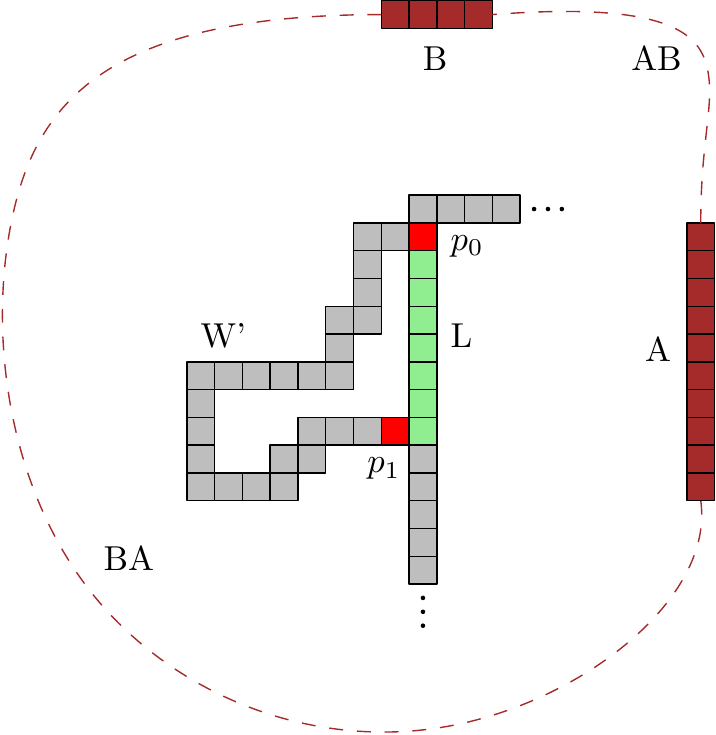}}
		\caption{A subpath $W'$ and its shortcut $L$ in green. To block $L$, $A$ and $B$ must exist. But then, either $p_0$ or $p_1$ (red tiles) will also be blocked. Therefore, also $W'$ cannot be built.}
		\label{fig:sequential}
	\end{figure}
	
	\begin{proof}
		Consider a sequentially constructible path $W$ and a subpath \remove{$W'$}\newlyrevised{$W'\subset W$} that  has a straight line $L$ connecting the startpoint and \newlyrevised{the} endpoint of $W'$.
		W.l.o.g., $L$ is a vertical line and we build from bottom to top.
		Assume that $(W\backslash W')\cup L$ is not constructible.
		Then at least two structures \revision{(which can be single tiles)} $A$ and $B$ must exist, preventing us from building $L$. Furthermore, these structures have to be connected via a path ($AB$ or $BA$, see Fig.~\ref{fig:sequential}).
		We observe that none of these connections can exist or otherwise, we cannot build $W$ \remove{(with $AB$ we cannot build the endpoint of $L$; with $BA$ we cannot build the first point of $W'$ to the left of $L$)}\newlyrevised{(if $AB$ exist, we cannot build the last tile $p_0$ of $L$; if $BA$ exist, we cannot build the first tile $p_1$ of $W'$)}. 
		Therefore, we can replace $W'$ with $L$.
	\end{proof}
	
	\revision{By repeating the construction of Lemma~\ref{lem:sequential} we get a shortest path from tile $t_1$ to $t_2$ in the following form:
		Let $P_1,\dots,P_k$ be reflex tiles on the path from $t_1$ to $t_2$. Furthermore, for every $1\leq i\leq k-1$, the path from $P_i$ to $P_{i+1}$ is monotone.
		This property holds for every shortest path, or else we can use shortcuts as in Lemma~\ref{lem:sequential}.}
	
	\begin{lemma}
		\label{lem:shortest}
		If a shortest path between two tiles is sequentially constructible, then every shortest path \newlyrevised{between these two tiles} is sequentially constructible.
	\end{lemma}
	
	\begin{proof}
		Consider a constructible shortest path $W$, a maximal subpath $W'$ that is $x$-$y$-monotone, and a bounding box $B$ around $W'$.
		Due to $L_1$-metric, any $x$-$y$-monotone path within $B$ is as long as $W'$.
		Suppose some path within $B$ is not constructible. Then we can use the same blocking argument as in Lemma~\ref{lem:sequential} to prove that \remove{also} $W'$ cannot be constructible \newlyrevised{as well}, contradicting that $W$ is constructible.
	\end{proof}
	
	\remove{\begin{theorem}
			We can check in polynomial time if there exists a (sequentially) constructible path from one tile in a simple polyomino to another.
		\end{theorem}
		\begin{proof}
			We search for one shortest path and check, whether it is sequential constructible.
			Combining Lemma~\ref{lem:sequential} and Lemma~\ref{lem:shortest} that a test for one shortest path is sufficient.
		\end{proof}}
		
		Using Lemma~\ref{lem:sequential} and Lemma~\ref{lem:shortest}, we are ready to prove Theorem~\ref{th:N2}.
		
		\medskip	
		{\bf Proof of Theorem~\ref{th:N2}}.
		Because it suffices to check one shortest path between two tiles, we can look at the BFS tree from each tile and then proceed like we did in Theorem~\ref{th:sequential:thin}. Thus, for each tile we perform a BFS in time $O(N)$ and a DFS with blocking look-ups in time $O(N\log N)$, which results in a total time of $O(N^2\log N)$.
		\qed

%% file: 5-short.tex
\section{Three-Dimensional Shapes}
\label{sec:3d-shapes}
An interesting and natural generalization of TAP is to consider three-dimensional shapes, i.e., polycubes.
The local considerations for simply connected two-dimensional shapes are no longer sufficient. In the following we show that deciding whether a polycube is constructible is \NP-hard.
\revised{Moreover, it is \NP-hard} to check whether there is a constructible path from a start cube $s$ to an end cube $t$ in a partial shape.

As a stepping stone, we start with a restricted version of the three-dimensional problem.

\remove{	If we are given a polycube we are able to show that it is \NP-hard to decide if it can be built even if we prohibit insertions from below (which is a likely scenario in practice).}
	\begin{theorem}
		\label{th:constr:polycube;hard1}
		It is \NP-hard to decide if a polycube can be built by inserting tiles only from above, north, east, south, and west.
	\end{theorem}
	
	\begin{figure}[h!]
		\centering
		\resizebox{0.9\textwidth}{!}{\includegraphics{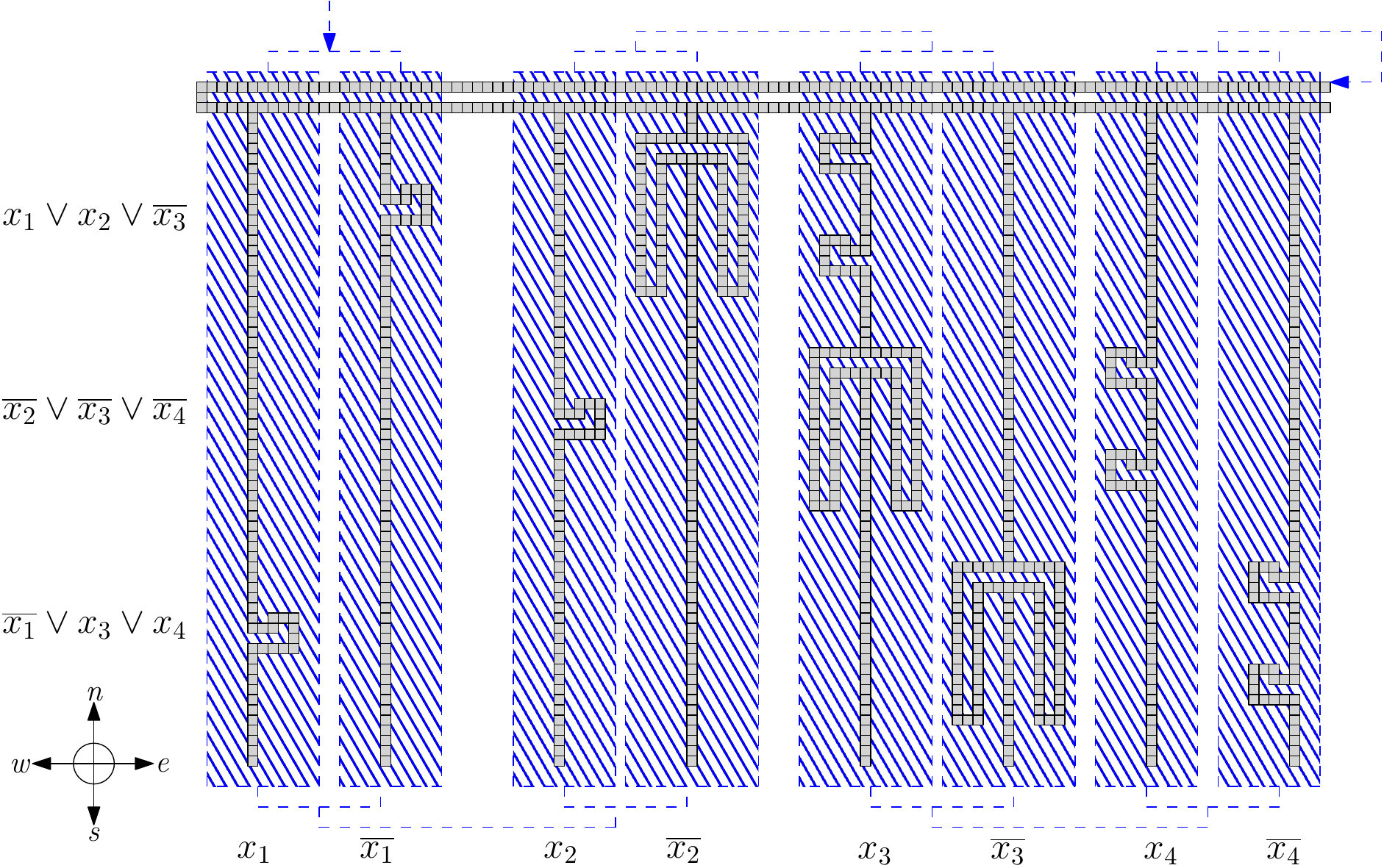}}
		\caption{	Top-view on the polycube. There is a vertical part going south for the \emph{true} and \emph{false} assignment of each variable.
			We start building at the top layer (blue) and have to block either the \emph{true} or the \emph{false} part of each variable from above.
			The blocked parts have to be built with only inserting from east, west, and south.
			For each clause, the parts of the inverted literals are modified to allow at most two of them being built in this way.
			All other parts can simply be inserted from above in the end.}
		\label{fig:constr:polycube:lowerlayer}
	\end{figure}
	\begin{figure}[h!]
		\begin{subfigure}[b]{0.6\textwidth}
			\centering
			\resizebox{0.9\columnwidth}{!}{\includegraphics{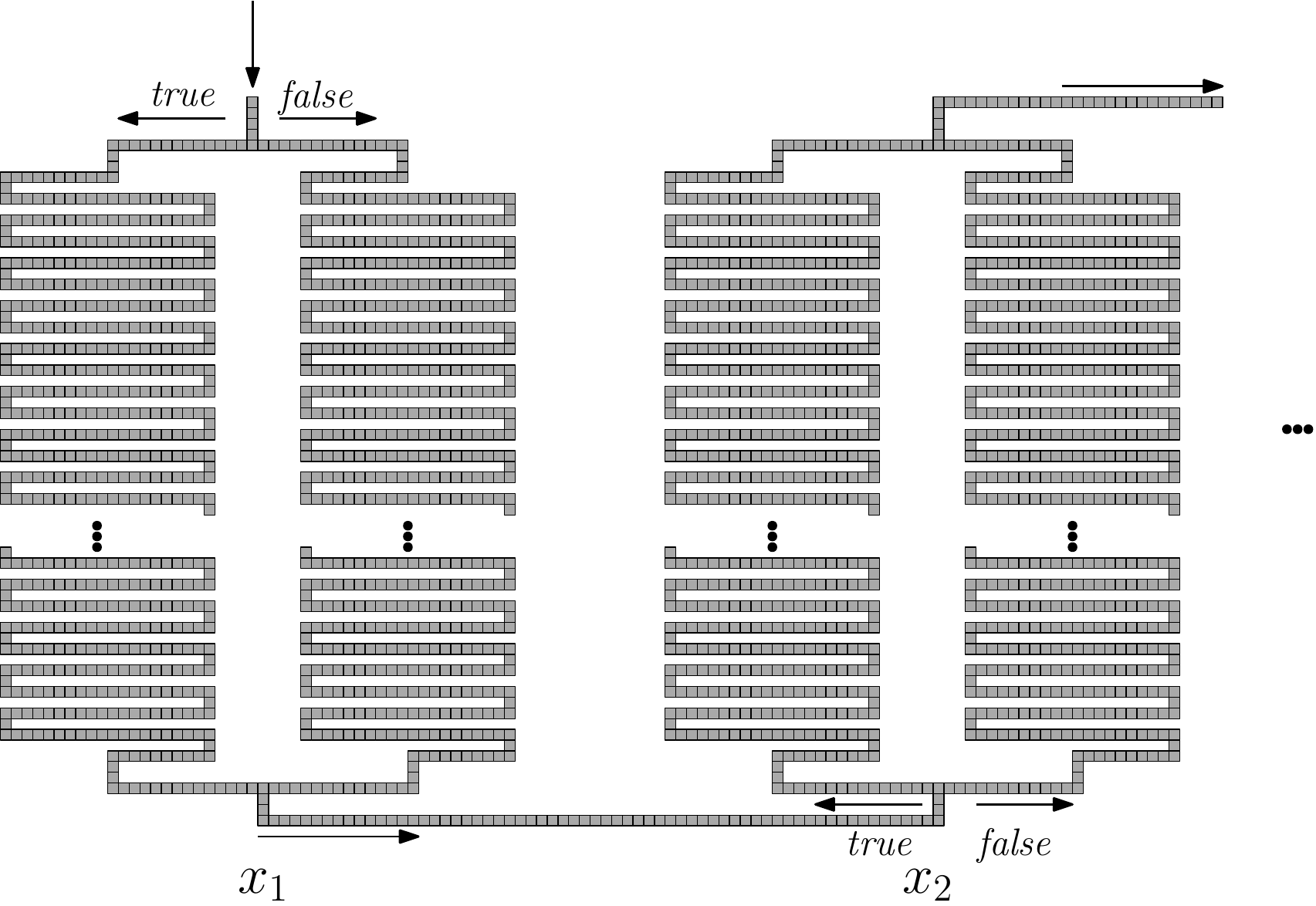}}
		\end{subfigure}\hfill
		\begin{subfigure}[b]{0.35\textwidth}
			\centering
			\resizebox{0.95\columnwidth}{!}{\includegraphics{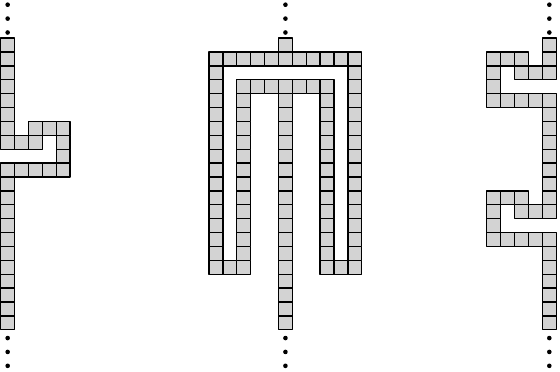}}
			\vspace{1cm}
		\end{subfigure}
		\caption{Top-view on the polycube. (Left)	In the beginning we have to block the access from the top for either the \emph{true} or \emph{false} part of the variable.
			The variable is assigned the blocked value. 
			(Right)	Three gadgets for a clause.
			Only two of them can be built if the tiles are only able to come from the east, south, and west.}
		\label{fig:constr:polycube}
	\end{figure}
	
	\begin{proof}
		We prove hardness by a reduction from \textsc{3SAT}.
		A visualization for the formula $(x_1\vee x_2 \vee \overline{x_3}) \wedge (\overline{x_2}\vee \overline{x_3}\vee \overline{x_4}) \wedge (\overline{x_1}\vee x_3 \vee x_4)$ can be seen in Fig.~\ref{fig:constr:polycube:lowerlayer}.
		It consists of two layers of interest (and some further auxiliary ones for space and forcing the seed tile by using the one-way gadget shown in Fig.~\ref{fig:oneway}).
		In the beginning, one has to build a part of the top layer (highlighted in blue in the example, details in Fig.~\ref{fig:constr:polycube}~(Right)).
		Forcing a specific start tile can be done by a simple construction.
		For each variable we have to choose to block the left (for assigning \emph{true}) or the right (for assigning \emph{false}) part of the lower layer.
		In the end, the remaining parts of the upper layer can trivially be filled from above.
		The blocked parts of the lower layer then have to be built with only inserting tiles from east, south, or west.
		In the end, the non-blocked parts can be filled in from above.
		For each clause we use a part (as shown in Fig.~\ref{fig:constr:polycube}~(Left)) that allows only at most two of its three subparts to be built from the limited insertion directions.
		We attach these subparts to the three variable values not satisfying the clause, i.e., the negated literals.
		This forces us to leave at least one negated literal of the clause unblocked, and thus at least one literal of the clause to be \emph{true}.
		Overall, this allows us to build the blocked parts of the lower layers only if the blocking of the upper level corresponds to a satisfying assignment.
		If we can build the \emph{true} and the \emph{false} parts of a variable in the beginning, any truth assignment for the variable is possible.
	\end{proof}
	
	\begin{figure}[h!]
		\centering
		\resizebox{0.28\textwidth}{!}{\includegraphics{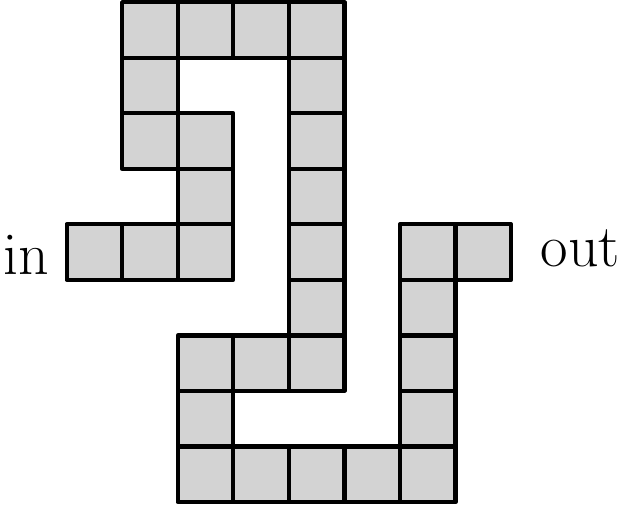}}
		\hfil
		\resizebox{0.52\textwidth}{!}{\includegraphics{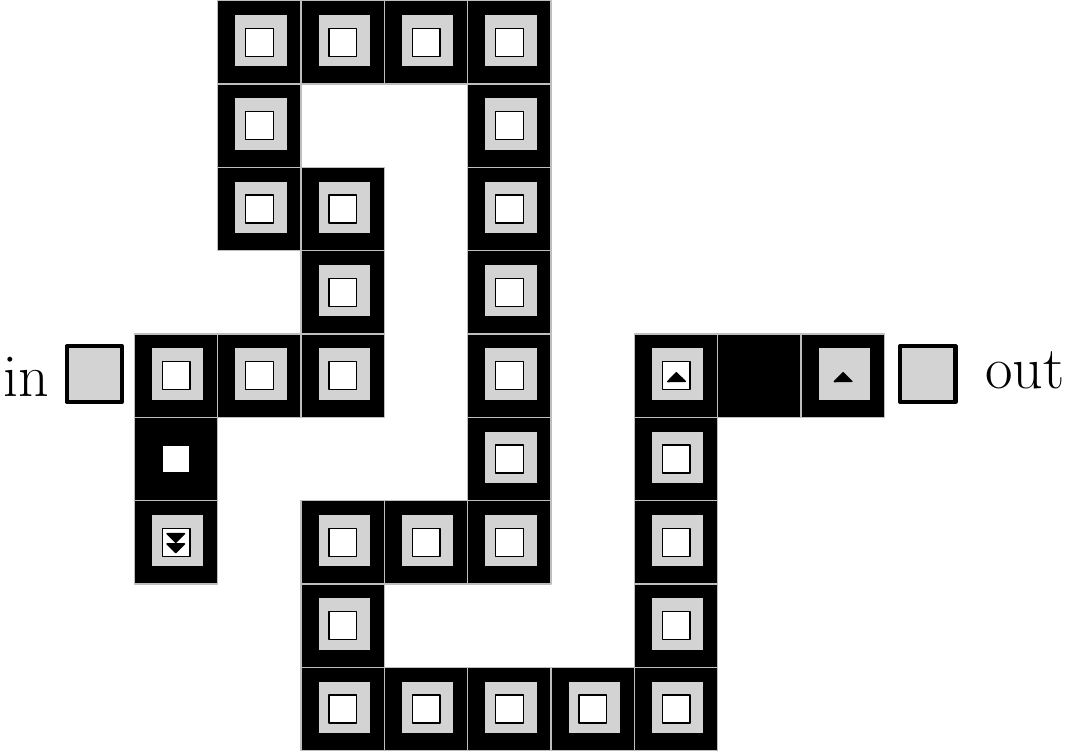}}
		\caption{(Left) This polyomino can only be constructed by
			starting at ``in'' and ending at ``out''. 
			(Right) \revision{By adding layers above (white) and below (black)
				this polyomino starting at the ``out''-tile, we obtain a polycube that is only
				constructible by starting at ``in'' (from the other direction we must build the black and white layer first and must then build the grey layer with 2D directions). Triangles denote where we can switch to another layer.} With this gadget we can enforce a seed
			tile.} \label{fig:oneway}
	\end{figure}

The construction can be extended to assemblies with arbitrary direction.

	\begin{theorem}
	\label{th:any}
		It is \NP-hard to decide if a polycube can be built by inserting tiles from any direction.
	\end{theorem}

	\begin{proof}
		We add an additional layer below the construction in Theorem~\ref{th:constr:polycube;hard1} that has to be built first and blocks access from below.
		Forcing the bottom layer to be built first can again be done with the one-way gadget shown in Fig.~\ref{fig:oneway}.
	\end{proof}

The difficulties of construction in 3D are highlighted by the fact that even identifying constructible
connections between specific positions is \NP-hard.

	\begin{theorem}
	\label{th:path}
		It is \NP-hard to decide whether a path from one tile to another can be built in a general polycube.
	\end{theorem}

\begin{figure}[h!]
	\centering
	\includegraphics[scale=0.5]{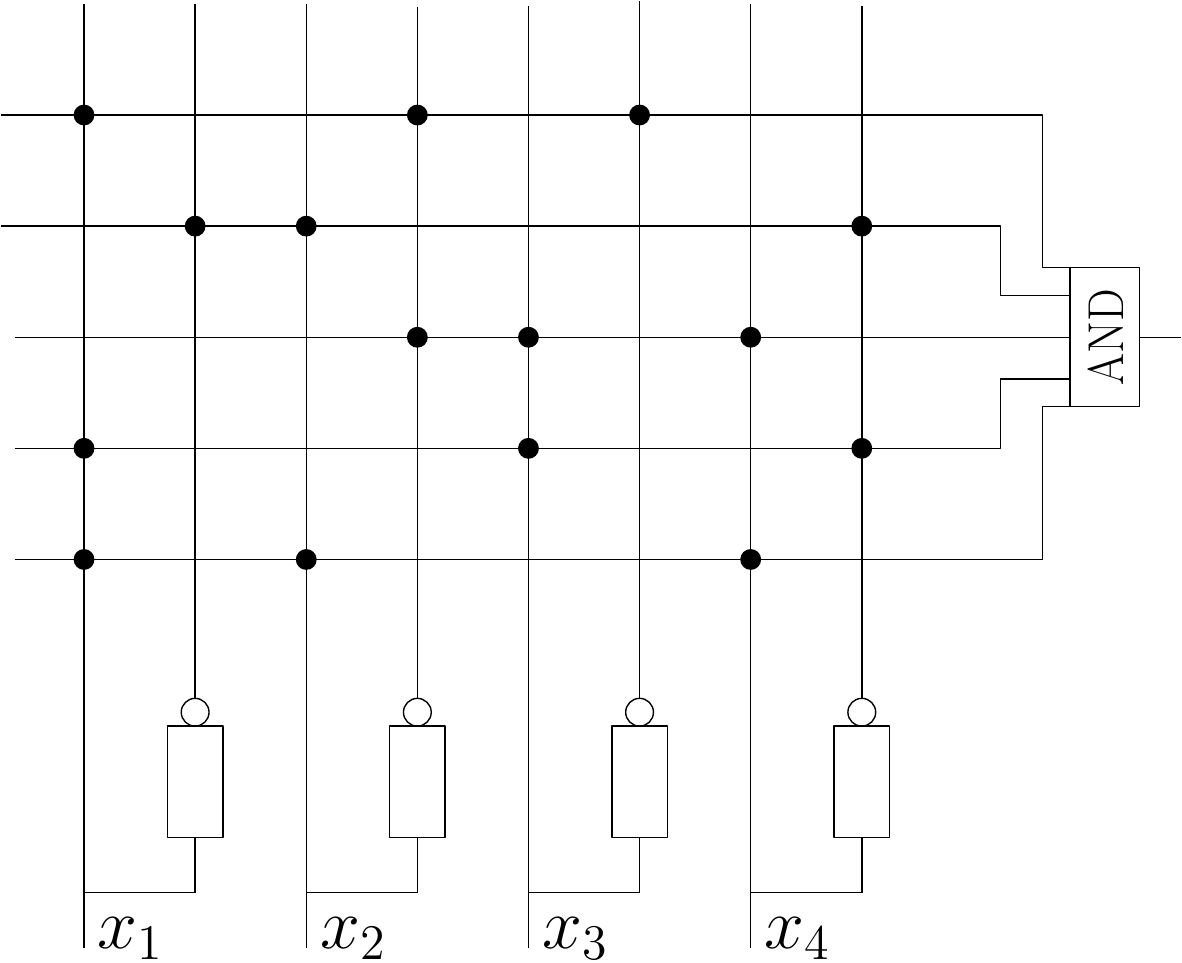}\hspace{1cm}
	\includegraphics[scale=0.5]{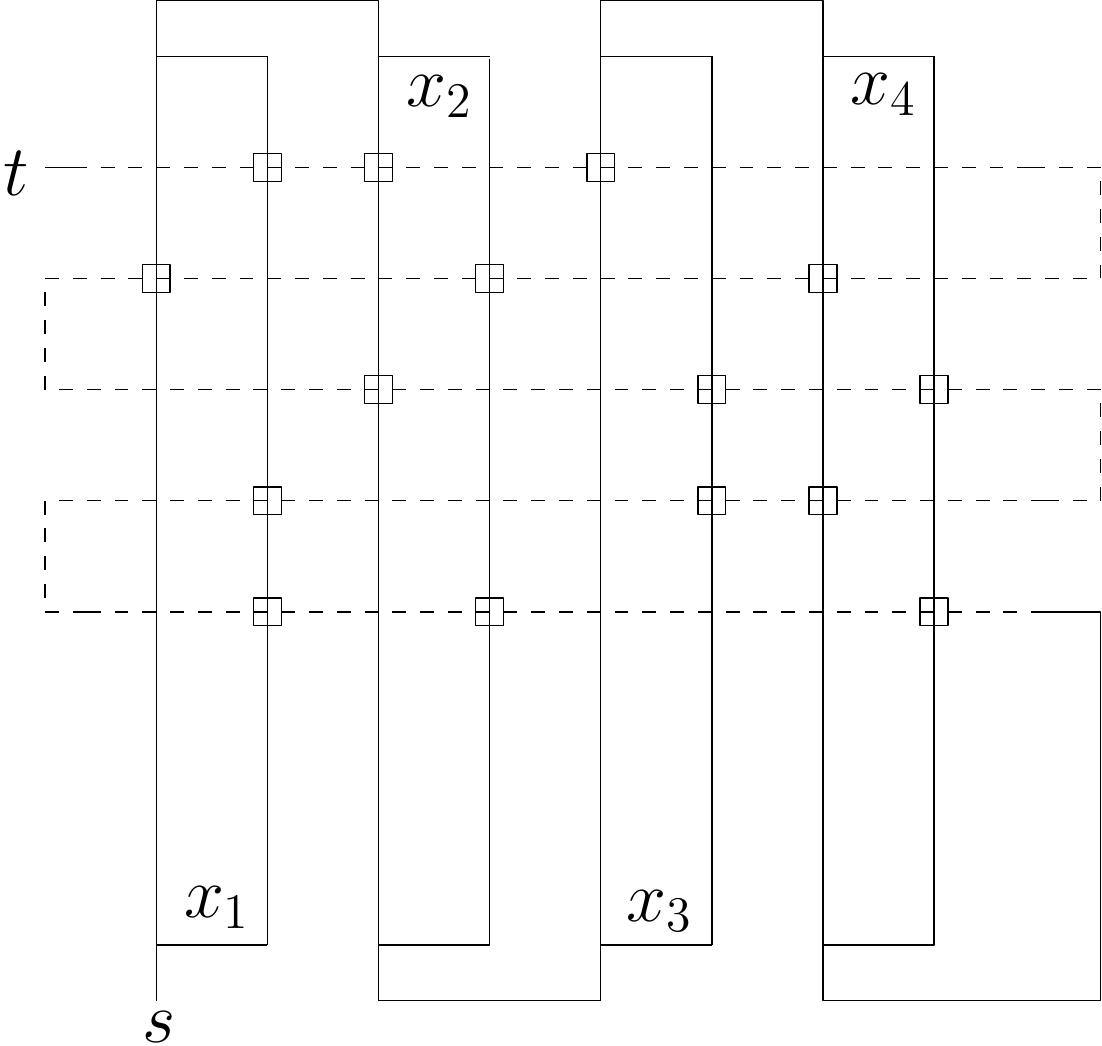}
	\caption{(Left) Circuit representation for the SAT formula $(x_1\vee \overline x_2\vee \overline x_3)\wedge (\overline x_1\vee x_2\vee \overline x_4) \wedge (\overline x_2\vee x_3\vee x_4)\wedge (x_1\vee x_3\vee \overline x_4)\wedge (x_1\vee  x_2\vee  x_4)$. (Right) Reduction from SAT formula. Boxes represent variable boxes.}
	\label{fig:3D:path:SAT}
\end{figure}

\begin{proof}
	
	We prove NP-hardness by a reduction from \textsc{SAT}.
	For each variable we have two vertical lines, one for the \emph{true} setting, one for the \emph{false} setting.
	Each clause gets a horizontal line and is connected with a variable if it appears as literal in the clause, see Fig~\ref{fig:3D:path:SAT}~(Left).
	We transform this representation into a tour problem where, starting at a point $s$, one first has to go through either the \emph{true} or \emph{false} line of each variable and then through all clause lines, see Fig.~\ref{fig:3D:path:SAT}~(Right).
	The clause part is only passable if the path in at least one crossing part (squares) does not cross, forcing us to satisfy at least one literal of a clause.
	As one has to go through all clauses, $t$ is only reachable if the selected branches for the variables equal a satisfying variable assignment for the formula.
	
	We now consider how to implement this as a polycube.
	The only difficult part is to allow a constructible clause path if there is a free crossing.
	In Fig.~\ref{fig:variablebox}~(Left), we see a variable box that corresponds to the crossing of the variable path at the squares in Fig.~\ref{fig:3D:path:SAT}~(Right).
	It blocks the core from further insertions.
	The clause path has to pass at least one of these variable boxes in order to reach the other side.
	See Fig.~\ref{fig:3D:path:SAT}~(Right) for an example.	
	Note that the corresponding clause parts can be built by inserting only from above and below, so there are no interferences.
\end{proof}

\begin{figure}[th]
	\centering
	\includegraphics[scale=0.3]{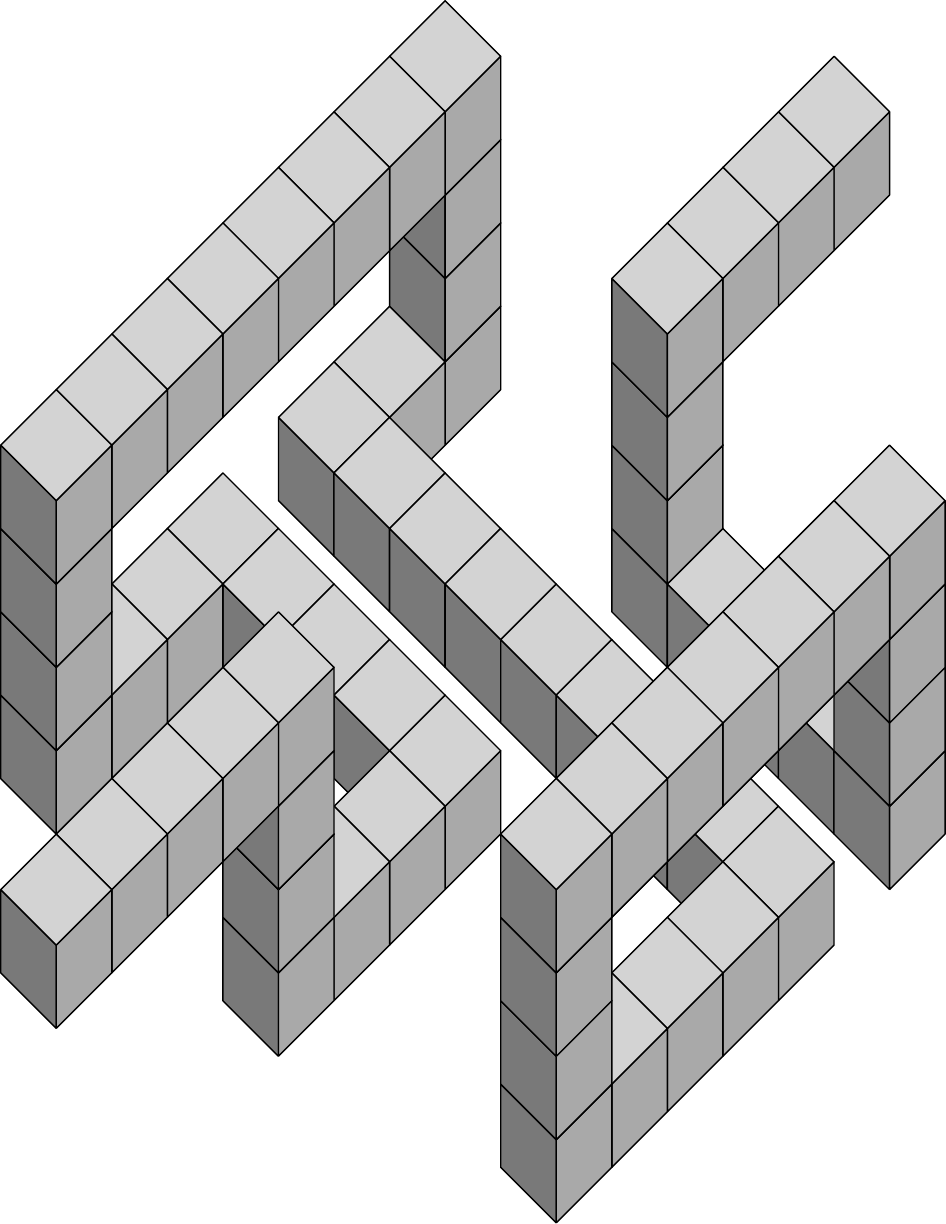}
	\hfil
	\resizebox{0.7\textwidth}{!}{\includegraphics{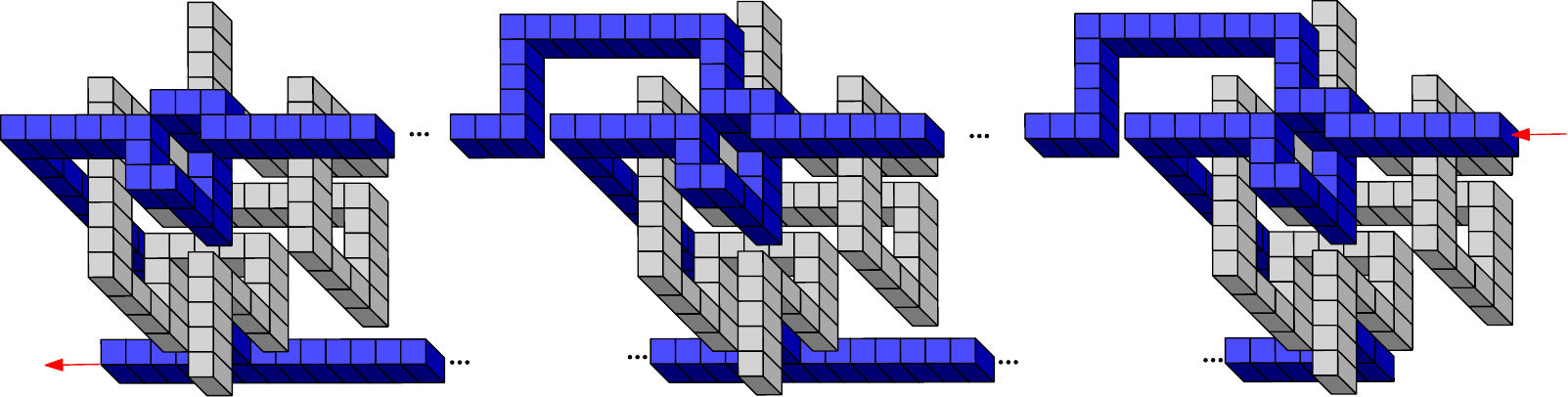}}
	\caption{(Left) Empty variable box. (Right) A clause line (blue) dips into a variable box. 
		If the variable box is built, then we cannot build the dip of the clause line.}
	\label{fig:variablebox}
\end{figure}

%% file: 06-future_work.tex
\section{Conclusion/Future Work}

We have provided a number of algorithmic results for Tilt Assembly. Various unsolved
challenges remain. What is the complexity of deciding TAP for non-simple
polyominoes? 
        While Lemma~\ref{lemma:constr:sop:removeconex} can be applied to all polyominoes, 
we cannot simply remove any convex tile. 
%
Can we find a constructible path in a polyomino from a given start and endpoint? This would help in finding a $\sqrt N$-approximation for non-simple polyominoes. How can we optimize the total makespan for 
constructing a shape?  And what options exist for non-constructible shapes?

An interesting approach may be to consider {\em staged} assembly, as shown in Fig.~\ref{fig:staged}, where
a shape gets constructed by putting together subpolyominoes, instead of adding one tile at a time.
This is similar to staged tile self-assembly~\cite{demaine2008staged,demaine2017new}.
This may also provide a path to sublinear assembly times, as a hierarchical assembly allows
massive parallelization. We conjecture that a makespan of $O(\sqrt{N})$ for a polyomino with $N$
tiles can be achieved.

All this is left to future work.

\begin{figure}[h]
	\centering
	\resizebox{0.2\textwidth}{!}{\includegraphics{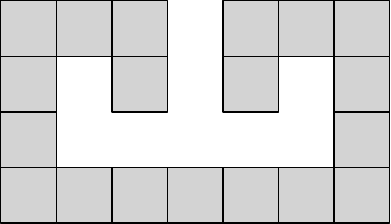}}
	\hspace{0.7cm}
	\resizebox{0.285\textwidth}{!}{\includegraphics{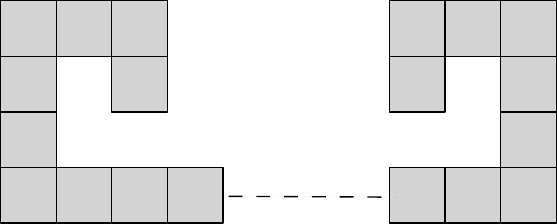}}
	\caption{(Left) A polyomino that cannot be constructed in the basic TAP model. (Right) Construction in a staged assembly model by putting together subpolyominoes.}
	\label{fig:staged}
\end{figure}

%% file: ms.bbl
\begin{thebibliography}{10}

\bibitem{Becker3810a}
A.~T. Becker, E.~D. Demaine, S.~P. Fekete, G.~Habibi, and J.~McLurkin.
\newblock Reconfiguring massive particle swarms with limited, global control.
\newblock In {\em Proceedings of the International Symposium on Algorithms and
  Experiments for Sensor Systems, Wireless Networks and Distributed Robotics
  ({ALGOSENSORS})}, pages 51--66, 2013.

\bibitem{becker2014particle}
A.~T. Becker, E.~D. Demaine, S.~P. Fekete, and J.~McLurkin.
\newblock Particle computation: Designing worlds to control robot swarms with
  only global signals.
\newblock In {\em {Proceedings IEEE International Conference on Robotics and
  Automation (ICRA)}}, pages 6751--6756, 2014.

\bibitem{becker2014crowdsourcing}
A.~T. Becker, C.~Ertel, and J.~McLurkin.
\newblock Crowdsourcing swarm manipulation experiments: A massive online user
  study with large swarms of simple robots.
\newblock In {\em {Proceedings IEEE International Conference on Robotics and
  Automation (ICRA)}}, pages 2825--2830, 2014.

\bibitem{becker2014simultaneously}
A.~T. Becker, O.~Felfoul, and P.~E. Dupont.
\newblock {Simultaneously powering and controlling many actuators with a
  clinical MRI scanner}.
\newblock In {\em Proceedings of the IEEE/RSJ International Conference on
  Intelligent Robots and Systems (IROS)}, pages 2017--2023, 2014.

\bibitem{becker2015toward}
A.~T. Becker, O.~Felfoul, and P.~E. Dupont.
\newblock {Toward tissue penetration by MRI-powered millirobots using a
  self-assembled Gauss gun}.
\newblock In {\em {Proceedings IEEE International Conference on Robotics and
  Automation (ICRA)}}, pages 1184--1189, 2015.

\bibitem{becker2013massive}
A.~T. Becker, G.~Habibi, J.~Werfel, M.~Rubenstein, and J.~McLurkin.
\newblock Massive uniform manipulation: Controlling large populations of simple
  robots with a common input signal.
\newblock In {\em Proceedings of the IEEE/RSJ International Conference on
  Intelligent Robots and Systems (IROS)}, pages 520--527, 2013.

\bibitem{berman1992complexity}
P.~Berman and G.~Schnitger.
\newblock On the complexity of approximating the independent set problem.
\newblock {\em Information and Computation}, 96(1):77--94, 1992.

\bibitem{cannon2012two}
S.~Cannon, E.~D. Demaine, M.~L. Demaine, S.~Eisenstat, M.~J. Patitz,
  R.~Schweller, S.~M. Summers, and A.~Winslow.
\newblock Two hands are better than one (up to constant factors).
\newblock In {\em Proc. Int. Symp. on Theoretical Aspects of Computer
  Science({STACS})}, pages 172--184, 2013.

\bibitem{chen2017parallelism}
H.-L. Chen and D.~Doty.
\newblock Parallelism and time in hierarchical self-assembly.
\newblock {\em SIAM Journal on Computing}, 46(2):661--709, 2017.

\bibitem{demaine2008staged}
E.~D. Demaine, M.~L. Demaine, S.~P. Fekete, M.~Ishaque, E.~Rafalin, R.~T.
  Schweller, and D.~L. Souvaine.
\newblock Staged self-assembly: nanomanufacture of arbitrary shapes with {O(1)}
  glues.
\newblock {\em Natural Computing}, 7(3):347--370, 2008.

\bibitem{demaine2017new}
E.~D. Demaine, S.~P. Fekete, C.~Scheffer, and A.~Schmidt.
\newblock New geometric algorithms for fully connected staged self-assembly.
\newblock {\em Theoretical Computer Science}, 671:4--18, 2017.

\bibitem{kim2015imparting}
P.~S.~S. Kim, A.~T. Becker, Y.~Ou, A.~A. Julius, and M.~J. Kim.
\newblock Imparting magnetic dipole heterogeneity to internalized iron oxide
  nanoparticles for microorganism swarm control.
\newblock {\em Journal of Nanoparticle Research}, 17(3):1--15, 2015.

\bibitem{kim2013swarm}
P.~S.~S. Kim, A.~T. Becker, Y.~Ou, M.~J. Kim, et~al.
\newblock {Swarm control of cell-based microrobots using a single global
  magnetic field}.
\newblock In {\em {Proceedings of the International Conference on Ubiquitous
  Robotics and Ambient Intelligence (URAI)}}, pages 21--26, 2013.

\bibitem{mahadev2016collecting}
A.~V. Mahadev, D.~Krupke, J.-M. Reinhardt, S.~P. Fekete, and A.~T. Becker.
\newblock Collecting a swarm in a grid environment using shared, global inputs.
\newblock In {\em Proc. IEEE Int. Conf. Autom. Sci. and Eng. (CASE)}, pages
  1231--1236, 2016.

\bibitem{manzoor2017parallel}
S.~Manzoor, S.~Sheckman, J.~Lonsford, H.~Kim, M.~J. Kim, and A.~T. Becker.
\newblock Parallel self-assembly of polyominoes under uniform control inputs.
\newblock {\em IEEE Robotics and Automation Letters}, 2(4):2040--2047, 2017.

\bibitem{shad2015particle}
H.~M. Shad, R.~Morris-Wright, E.~D. Demaine, S.~P. Fekete, and A.~T. Becker.
\newblock Particle computation: Device fan-out and binary memory.
\newblock In {\em {Proceedings IEEE International Conference on Robotics and
  Automation (ICRA)}}, pages 5384--5389, 2015.

\bibitem{Shahrokhi2015}
S.~Shahrokhi and A.~T. Becker.
\newblock Stochastic swarm control with global inputs.
\newblock In {\em Proceedings of the IEEE/RSJ International Conference on
  Intelligent Robots and Systems (IROS)}, pages 421--427, 2015.

\bibitem{winfree1998algorithmic}
E.~Winfree.
\newblock {\em Algorithmic self-assembly of DNA}.
\newblock PhD thesis, California Institute of Technology, 1998.

\end{thebibliography}
